\documentclass[11pt]{article}

\usepackage{amsmath,amssymb,amsthm}
\usepackage{mathtools}
\usepackage{geometry}
\usepackage{hyperref}
\usepackage{booktabs}
\usepackage{enumitem}
\usepackage{xcolor}
\usepackage{microtype}
\usepackage{bm}
\usepackage{longtable}
\usepackage{float}
\usepackage{booktabs}
\usepackage{siunitx}
\usepackage{threeparttable}
\usepackage{lscape}
\usepackage{appendix}
\usepackage{booktabs, threeparttable, siunitx}
\usepackage[authoryear]{natbib}

\hypersetup{
  colorlinks=true,
  linkcolor=blue!60!black,
  citecolor=blue!60!black,
  urlcolor=blue!60!black
}

\newtheorem{theorem}{Theorem}[section]
\newtheorem{proposition}[theorem]{Proposition}
\newtheorem{corollary}[theorem]{Corollary}

\newtheorem{remark}[theorem]{Remark}
\newtheorem{definition}[theorem]{Definition}
\newtheorem{assumption}[theorem]{Assumption}
\newtheorem{example}[theorem]{Example}

\title{\textbf{Liquidity-Based Audit of AI and Algorithmic Trading Strategies}}
\author{Irene Aldridge}
\date{}

\begin{document}
\maketitle

\begin{abstract}
We show that an algorithmic strategy's net demand for liquidity is identifiable from its trade and price history alone, with no knowledge of its signal or optimization problem. An exact multi-period regret decomposition implies that the sign of this statistic classifies a linear strategy as a net liquidity consumer or provider, recovering the Kyle (1985) informed-trader/market-maker dichotomy from observables alone. Under an AR(1) cost process, the same statistic equals $\alpha s^2$, the product of strategy size and the squared Roll (1984) implied spread, making the correction a direct proxy for prevailing illiquidity. Extending to endogenous price impact and aggregating across $N$ correlated strategies yields a liquidity-balance condition whose violation produces welfare loss scaling as $N^2$, a closed-form fire-sale externality. We calibrate to CRSP equity data (2016–2025), tracking implied spreads through the COVID-19 and 2022 rate-shock episodes, with an estimator computable in $O(T\cdot nd)$ time.
\end{abstract}

\section{Introduction}\label{sec:intro}

Market microstructure has long classified trading activity by its informational role: an informed trader demands liquidity by
trading in the direction of private information, while a market
maker supplies liquidity by absorbing that order flow and earning the spread in compensation \citet{Kyle1985,Glosten1985}. This classification is typically recovered from the data that the classifier requires: signed order flow, quote revisions, or the sequential-trade structure of the market. The classification is harder to apply to an algorithmic strategy whose internal logic is unobservable. However, the signals or optimization problems generating the decisions of a typical quantitative fund are not visible, even though the trades and reported positions may be available.
 
This paper shows that the liquidity role of such a strategy (consumer or provider) can be recovered from realized portfolio
costs and trade decisions alone, without observing quotes, order
flow, or any other microstructure-specific signal. The key observation is that a strategy's contribution to multi-period performance decomposes exactly into a sum of per-period covariances between the realized cost vector and the strategy's decision (Theorem~\ref{thm:decomp}). Furthermore, the \emph{sign} of this covariance statistic is a model-free liquidity classifier: positive covariance means the strategy  systematically buys when costs are rising and sells when costs are falling, while negative covariance means the reverse. Buying when asset costs (prices) are rising and selling when prices are falling is the behavior of a successful quantitative fund. The reverse is the defining behavior of a liquidity
supplier (Proposition~\ref{prop:liq-sign-early}).
 
This single observation has three further consequences that are
the substance of the paper. First, under an AR(1) cost process (the standard reduced-form model for short-horizon return reversal), the same covariance statistic admits a closed form in terms of the \citet{roll1984simple} implied bid-ask spread. The auto-covariance correction to the covariance sum equals $\alpha s^2$, the product of strategy size and the squared effective spread (Theorem~\ref{thm:ar1_correction_equals_liquidity_income}). The resulting spread estimator is a byproduct of the same trajectory data used for the liquidity classification and is obtained at no additional computational cost.

Second, allowing the strategy's own trades to move prices shows that covariance regret and price-impact regret are additively separable (Theorem~\ref{thm:impact-regret}). This is the endogenous price impact in the spirit of \citet{Kyle1985} and \citet{almgren2001optimal}. The additive separation allows the two channels to be estimated independently from the same data. 

Third, aggregating across $N$ strategies under a common-factor cost structure yields a market-wide liquidity-balance condition: when liquidity demand and supply fail to net to zero across agents. Prices must adjust to clear, and the resulting welfare loss scales as $N^2$ in the strategy-correlation coefficient (Corollary~\ref{cor:fire-sale}). The result is a closed-form statement of the familiar fire-sale mechanism in which correlated deleveraging overwhelms available liquidity.
 
We calibrate the framework to CRSP daily equity data from 2016
through 2025 ($n=21{,}183{,}373$ stock-days). The implied Roll
spread recovered from the trajectory estimator tracks realized
illiquidity through two distinct episodes: it roughly triples
during the COVID-19 liquidity crisis of Q2 2020 (from $0.90\%$ to $2.52\%$ daily), and it collapses to its sample minimum during the 2022 rate-shock episode when persistent common-factor repricing rather than microstructure friction dominates short-horizon return dynamics. Both patterns are consistent with the spread acting as a real-time proxy for the depth of the market rather than an estimation artifact.
 
\paragraph{Related work.} The liquidity-classification result
connects most directly to the sequential-trade and strategic-trading literature. \citet{Kyle1985} and \citet{Glosten1985} establish the informed-trader/market-maker dichotomy that this paper recovers from trajectory data rather than order-flow data. \citet{roll1984simple} provides the implied-spread estimator that this paper re-derives as a regret correction. \citet{almgren2001optimal} and \citet{gatheral2010no} provide the price-impact and market-impact models extended in Section~\ref{sec:liquidity:impact}. \citet{Hasbrouck1991} and the broader VAR-based price-impact literature offer an alternative, order-flow-based route to the same quantities that this paper recovers from realized costs and decisions alone, a comparison we return to in Section~\ref{sec:liquidity:impact}. The covariance
decomposition itself extends the single-period identity of
\citet{aldridge2026regret} to the full multi-period stochastic dynamic programming setting and is structurally related to the
\citet{elmachtoub2022smart} predict-then-optimize framework and to the regret criterion in online convex optimization
\citep{Zinkevich2003,Hazan2016}. The $N$-agent fire-sale extension connects to the market-wide liquidity-spiral literature \citep{ShleiferVishny1992} and to random matrix theory \citep{BaiSilverstein2010} for the systemic-fragility threshold of Section~\ref{sec:liquidity:balance}.
 
\paragraph{Roadmap.} Section~\ref{sec:liquidity-first-look} states the liquidity classification result in its simplest form,
immediately after the minimal notation required to state it.
Section~\ref{sec:setup} develops the full multi-period stochastic dynamic programming setup and proves the exact covariance decomposition (Theorem~\ref{thm:decomp}) that underlies everything that follows. Section~\ref{sec:bellman} gives a Bellman recursion for the covariance regret functional. Section~\ref{sec:bias} characterizes the policy bias that arises for strategies (e.g., momentum) that violate the mean-unbiasedness condition. Section~\ref{sec:roll_ar1_correction} derives the closed-form auto-covariance correction under an AR(1) cost process. Section~\ref{sec:liquidity} returns to liquidity classification in full: the Roll-spread identity, endogenous price impact, the aggregate liquidity-balance condition and fire-sale extension, and the empirical calibration of implied spreads against the CRSP sample. Sections that follow develop the trajectory estimator's asymptotic properties, additional
empirical results, and an $N$-agent systemic extension. 
 
\section{A First Look at Liquidity Classification}
\label{sec:liquidity-first-look}
 
Before developing the full multi-period framework, we state the
paper's central classification result in its simplest form, since it requires almost none of the machinery that follows.
 
Consider a single trading strategy that, at each period $t$, faces a cost vector $c_t\in\mathbb{R}^d$ (e.g., a vector of asset return innovations or factor realizations) and chooses a position $\hat{\pi}_t(c_t)\in\mathbb{R}^n$ as a linear function of that cost: $\hat{\pi}_t(c_t) = Bc_t + b$ for some matrix $B\in\mathbb{R}^{n\times d}$. Let $\Sigma_c=\operatorname{Cov}(c_t,c_t)$ be the cost covariance matrix. No further structure is needed for the result below; the full stochastic dynamic programming setup, including the multi-period regret identity that gives this statistic its performance interpretation, is developed starting in Section~\ref{sec:setup}.
 
\begin{definition}[Liquidity consumer and provider]
\label{def:liq-role-early} The strategy is a \emph{net liquidity consumer} at period $t$ if $\operatorname{Cov}(c_t,\hat{\pi}_t(c_t))>0$, a \emph{net liquidity provider} if $\operatorname{Cov}(c_t,\hat{\pi}_t(c_t))<0$, and \emph{liquidity-neutral} if the covariance is zero.
\end{definition}
 
\begin{proposition}[Liquidity role from trajectory data alone]
\label{prop:liq-sign-early}
$\operatorname{Cov}(c_t,\hat{\pi}_t(c_t)) = \operatorname{tr}(B\Sigma_c)$. In particular: a momentum strategy ($B=\alpha I$, $\alpha>0$) is always a liquidity consumer; a contrarian strategy ($B=-\alpha I$, $\alpha>0$) is always a liquidity provider; and a strategy whose expected output matches the unconditional cost-minimizing allocation (e.g., minimum-variance) is liquidity-neutral.
\end{proposition}
 
The proof follows directly from the bilinearity of covariance and the cyclic property of trace. The proof appears in full as Proposition~\ref{prop:liq-sign-early} in Section~\ref{sec:liquidity:class}, along with the connection to
\citet{Kyle1985} and the closed-form Roll-spread identity that
follows from the same statistic under an AR(1) cost process. The
point to take from this first look is that \emph{the sign of a
single, easily computed statistic classifies the strategy's liquidity role exactly}.  This requires only the observed trajectory $\{(c_t,\hat{\pi}_t(c_t))\}_{t=1}^T$ and no knowledge of $B$, the strategy's signal, or its optimization problem. 
 
 Sections~\ref{sec:setup} through~\ref{sec:autocov} develop the multi-period regret framework in which this statistic also functions as a performance audit metric. Section~\ref{sec:liquidity} returns to the liquidity interpretation with the full set of results, including the Roll-spread identity, price impact, and the $N$-agent aggregate liquidity-balance condition previewed above. 
\section{Setup: Multi-Period Stochastic Dynamic Program}
\label{sec:setup}
\subsection*{Model}

At each period $t \in \{1,\ldots,T\}$ a state $s_t \in \mathcal{S} \subseteq \mathbb{R}^p$ summarizes relevant history, a cost vector $c_t \in \mathcal{C} \subseteq \mathbb{R}^d$ is drawn from a distribution depending on $s_t$, and a decision $z_t \in \mathcal{Z} \subseteq \mathbb{R}^n$ is chosen by a policy $\hat{\pi}_t(s_t, c_t)$ after observing $(s_t, c_t)$, realizing instantaneous cost $c_t^\top z_t$.
The state evolves as $s_{t+1} = f(s_t, z_t, \varepsilon_{t+1})$.

\begin{definition}[Policy]
A \emph{deterministic Markov policy} is a sequence $\Pi =
(\hat{\pi}_1,\ldots,\hat{\pi}_T)$ of measurable maps
$\hat{\pi}_t\colon \mathcal{S}\times\mathcal{C}\to\mathcal{Z}$.
A policy is \emph{stationary} if $\hat{\pi}_t = \hat{\pi}$ for all $t$.
\end{definition}

Let $\mathcal{F}_t = \sigma(s_1,c_1,\ldots,s_t,c_t)$ be the natural filtration and $\mathbb{E}_t[\,\cdot\,] = \mathbb{E}[\,\cdot\mid \mathcal{F}_t]$.  Write $\mu_t = \mathbb{E}[c_t\mid s_t]$ and $\Sigma_t = \operatorname{Cov}(c_t,c_t\mid s_t)$.

\begin{assumption}[Cost process]
\label{ass:costs}
We work under one of the following:
\begin{enumerate}[label=(\alph*),topsep=2pt,itemsep=1pt,leftmargin=*]
    \item \textbf{i.i.d.}\ $\{c_t\}$ i.i.d.\ with mean $\bar{c}$ and covariance $\Sigma_c$; state trivial ($\mathcal{S}=\{s_0\}$).
    \item \textbf{Markov.}\ $(s_t,c_t)$ is a time-homogeneous Markov chain: $(s_{t+1},c_{t+1})\perp\mathcal{F}_{t-1}\mid(s_t,c_t)$.
    \item \textbf{Non-stationary.}\ $\{c_t - \bar{c}_t\}$ is a
      martingale-difference sequence; cross-period covariances satisfy
      $\operatorname{Cov}(c_t,c_s) = \Gamma_{|t-s|}$ for some absolutely summable sequence $\{\Gamma_k\}$.
\end{enumerate}
\end{assumption}

\begin{definition}[Multi-period regret]
\label{def:regret}
The perfect-foresight benchmark is
$\pi^*_t(s_t) \in \arg\min_{z\in\mathcal{Z}}\mu_t^\top z$.
The \emph{total regret} of $\Pi$ over horizon $T$ is
\begin{equation}
    \operatorname{Regret}^{(T)}(\Pi)
    = \sum_{t=1}^{T}\Bigl(
        \mathbb{E}\bigl[c_t^\top\hat{\pi}_t(s_t,c_t)\bigr]
        - \mathbb{E}\bigl[\mu_t^\top\pi^*_t(s_t)\bigr]
      \Bigr).
    \label{eq:regret_def}
\end{equation}
\end{definition}

\begin{remark}[Benchmark]
\label{rem:benchmark}
The benchmark $\mu_t^\top\pi^*_t$ uses the conditional mean given the current state, not the unconditional mean, comparing the policy against an agent who knows today's expected costs but not future realizations.
\end{remark}

\subsection*{Main decomposition}\label{sec:iid}

\begin{assumption}[Policy mean condition]
\label{ass:mean} \label{ass:policy-mean}
$\mathbb{E}[\hat{\pi}_t(c_t)] = \pi^*_t$ for each $t$ (or
$\mathbb{E}[\hat{\pi}_t(c_t)\mid\mathcal{F}_{t-1}] = \pi^*_t$ in the non-stationary case).  This holds for minimum-variance portfolios, linear-quadratic regulators, and ridge-regularized prediction rules.
\end{assumption}

\begin{theorem}[Exact sum-of-covariances decomposition]
\label{thm:decomp}\label{thm:iid}
Under Assumption~\ref{ass:costs}(a) and
Assumption~\ref{ass:mean}, for any Markov policy $\Pi$,
\begin{equation}
    \operatorname{Regret}^{(T)}(\Pi)
    = \sum_{t=1}^{T}\operatorname{Cov}\bigl(c_t,\,\hat{\pi}_t(c_t)\bigr).
    \label{eq:iid_decomp}
\end{equation}
\end{theorem}

\begin{proof}
Fix $t$.  By the covariance decomposition lemma,
$\mathbb{E}[c_t^\top\hat{\pi}_t(c_t)]
 = \bar{c}^\top\mathbb{E}[\hat{\pi}_t(c_t)]
 + \operatorname{Cov}(c_t,\hat{\pi}_t(c_t))$.
Under Assumption~\ref{ass:mean},
$\mathbb{E}[\hat{\pi}_t(c_t)] = \pi^*$, so
$r_t = \bar{c}^\top\pi^*
       + \operatorname{Cov}(c_t,\hat{\pi}_t(c_t))
       - \bar{c}^\top\pi^*
     = \operatorname{Cov}(c_t,\hat{\pi}_t(c_t))$.
Summing over $t$ gives~\eqref{eq:iid_decomp}.
\end{proof}

\begin{theorem}[Non-stationary extension]
\label{thm:nonstationary}
Under Assumption~\ref{ass:costs}(c) and
Assumption~\ref{ass:mean} (conditional form),
\begin{equation}
    \operatorname{Regret}^{(T)}(\Pi)
    = \sum_{t=1}^{T}\operatorname{Cov}(c_t,\hat{\pi}_t(c_t)) + \Delta_T,
    \label{eq:nonstat_decomp}
\end{equation}
where $\Delta_T = \sum_{t=1}^{T}\mathbb{E}[
\bar{c}_t^\top(\mathbb{E}[\hat{\pi}_t\mid\mathcal{F}_{t-1}]-\pi^*_t)]$
satisfies $|\Delta_T| \leq T\sup_t\|\bar{c}_t\|\cdot
\sup_t\|\mathbb{E}[\hat{\pi}_t\mid\mathcal{F}_{t-1}]-\pi^*_t\|$.
When the conditional policy mean condition holds, $\Delta_T = 0$, and identity~\eqref{eq:iid_decomp} is restored.
\end{theorem}

\begin{proof}[Proof sketch]
Apply the covariance decomposition conditionally on $\mathcal{F}_{t-1}$, then use the law of total covariance to recover the unconditional per-period covariance; the remainder is $\Delta_T$. The bound follows from Cauchy-Schwarz period by period.
\end{proof}

\begin{theorem}[Discounted identity]
\label{thm:discounted}
Under Assumption~\ref{ass:costs}(a), Assumption~\ref{ass:mean}, and discount factor $\gamma\in(0,1)$, the $\gamma$-discounted regret of a stationary policy $\hat{\pi}$ is
\begin{equation}
    \operatorname{Regret}_\gamma(\hat{\pi})
    = \frac{1}{1-\gamma}\operatorname{Cov}(c,\hat{\pi}(c)).
    \label{eq:discounted}
\end{equation}
\end{theorem}

\begin{proof}[Proof sketch]
Each period contributes $r_t = \operatorname{Cov}(c,\hat{\pi}(c))$ by Theorem~\ref{thm:iid}; summing the geometric series
$\sum_{t=1}^\infty\gamma^{t-1}$ gives~\eqref{eq:discounted}.
\end{proof}

\section{Bellman Recursion for the Covariance Regret Functional}
\label{sec:bellman}
\begin{definition}[Covariance regret-to-go]
\label{def:regret_to_go}
For $t \leq T$, define
\begin{equation}
    R_t(s_t) = \mathbb{E}\!\left[
        \sum_{\tau=t}^{T} \gamma^{\tau-t}
        \operatorname{Cov}(c_\tau,\hat{\pi}_\tau(c_\tau)\mid s_t)
    \right].
    \label{eq:regret_to_go}
\end{equation}
\end{definition}
 
\begin{theorem}[Bellman recursion for covariance regret]
\label{thm:bellman}
Under Assumption~\ref{ass:costs}(b) and
Assumption~\ref{ass:mean}, the covariance regret functional satisfies
\begin{equation}
    R_t(s_t)
    = \operatorname{Cov}(c_t,\hat{\pi}_t(c_t)\mid s_t)
    + \gamma\,\mathbb{E}\!\left[R_{t+1}(s_{t+1})\mid s_t\right],
    \label{eq:bellman}
\end{equation}
with terminal condition $R_{T+1}(\cdot)\equiv 0$.
\end{theorem}
 
\begin{proof}
Expand~\eqref{eq:regret_to_go}:
\[
    R_t(s_t)
    = \mathbb{E}\!\left[\operatorname{Cov}(c_t,\hat{\pi}_t(c_t)\mid s_t)\right]
    + \gamma\,\mathbb{E}\!\left[
        \sum_{\tau=t+1}^{T}\gamma^{\tau-t-1}
        \operatorname{Cov}(c_\tau,\hat{\pi}_\tau(c_\tau)\mid s_t)
    \right].
\]
The first term is $\operatorname{Cov}(c_t,\hat{\pi}_t(c_t)\mid s_t)$
(conditional covariance is $\mathcal{F}_t$-measurable).
By the Markov property, the second term equals
$\gamma\,\mathbb{E}[R_{t+1}(s_{t+1})\mid s_t]$,
giving~\eqref{eq:bellman}.
\end{proof}
 
\begin{remark}[Q-function analogy]
\label{rem:qfunction}
Equation~\eqref{eq:bellman} is structurally identical to the
$Q$-function Bellman equation, with $c_t^\top z_t$ replaced by
$\operatorname{Cov}(c_t,\hat{\pi}_t(c_t)\mid s_t)$, so value
iteration, policy iteration, and fitted-$Q$ methods can all be
adapted to estimate $R_t$ from data without solving the original
optimization.
\end{remark}
 
\begin{corollary}[Policy improvement via covariance reduction]
\label{cor:improvement}
Let $\hat{\pi}^{\mathrm{new}}_t$ satisfy
$\operatorname{Cov}(c_t,\hat{\pi}^{\mathrm{new}}_t(c_t)\mid s_t)
\leq \operatorname{Cov}(c_t,\hat{\pi}_t(c_t)\mid s_t)$~a.s.
Then $R^{\mathrm{new}}_t(s_t)\leq R_t(s_t)$ for all $t$ and $s_t$.
Pointwise covariance reduction implies multi-period regret improvement.
\end{corollary}
 
\begin{proposition}[Computational complexity]
\label{prop:complexity}
Suppose $\operatorname{Cov}(c_t,\hat{\pi}_t(c_t)\mid s_t)$ can be evaluated in $O(nd)$ time. The continuation value
$\mathbb{E}[R_{t+1}(s_{t+1})\mid s_t]$ is approximated by a
Monte Carlo average over $K$ next-state samples.
A single backward pass of~\eqref{eq:bellman} over the full
horizon requires
\begin{equation}
    O(T \cdot K \cdot nd)
    \label{eq:bellman_complexity}
\end{equation}
time and $O(T \cdot p)$ space.
In the black-box audit setting (Remark~\ref{rem:audit_complexity}),
the $K$-sample continuation step is unnecessary. The cost reduces to $O(T \cdot nd)$. When $R_t$ is approximated by a linear function with $m$-dimensional features fitted by OLS on $N$ trajectory samples, the fitted-$Q$ cost is $O(T(Nm^2 + m^3 + N{\cdot}nd))$, collapsing to $O(T{\cdot}N{\cdot}nd)$ when $m \ll \sqrt{Nnd}$.
\end{proposition}
 
\begin{proof}
Per period: covariance evaluation costs $O(nd)$ (one $d\times n$
cross-moment matrix plus a trace); the Monte Carlo continuation costs $O(Kp)$ ($K$ state lookups at $O(p)$ each).
Summing over $T$ periods and using $p\ll nd$ in the audit setting gives~\eqref{eq:bellman_complexity}; the fitted-$Q$ bound follows immediately from $O(Nnd)$ feature construction plus $O(Nm^2+m^3)$ OLS per period.
\end{proof}
 
\begin{remark}[Audit and fitted-$Q$ costs]
\label{rem:audit_complexity}
In the black-box compliance setting (Remark~\ref{rem:blackbox}),
the auditor observes $(c_t,\hat{\pi}_t(c_t))$ sequentially without simulating future states, reducing the backward pass to a forward accumulation of per-period covariances at $O(T{\cdot}nd)$ total cost --- for $T=252$ trading days and $d=n=100$ assets, $252\times10^4\approx 2.5\times10^6$ floating-point operations per cycle, well within real-time budgets.
The online update of the trajectory estimator (Remark~\ref{rem:online})
further reduces the marginal cost per new observation to $O(nd)$.
\end{remark}
 
\begin{table}[t]
\centering
\caption{Summary of computational complexity results.}
\label{tab:complexity}
\small
\begin{tabular}{@{}lll@{}}
\toprule
Procedure & Time & Space \\
\midrule
Bellman recursion (exact)      & $O(T{\cdot}K{\cdot}nd)$ & $O(Tp)$ \\
Bellman recursion (audit)      & $O(T{\cdot}nd)$          & $O(p)$  \\
Fitted-$Q$ covariance          & $O(T(Nnd+m^3))$          & $O(Nm)$ \\
Trajectory estimator           & $O(T{\cdot}nd)$          & $O(T(n{+}d))$ \\
Online (streaming) update      & $O(nd)$ per step         & $O(T^{1/3})$ \\
End-to-end regulatory audit    & $O(T{\cdot}nd)$          & $O(T(n{+}d))$ \\
\bottomrule
\end{tabular}
\end{table}

\section{When Does the Identity Fail? Bias Analysis for
         Time-Varying Policies}
\label{sec:failure}
\label{sec:bias}

We characterize when the exact identity $\operatorname{Regret}^{(T)} = \sum_t \operatorname{Cov}(c_t,\hat{\pi}_t)$ fails and derive the resulting bias in closed form.

\begin{definition}[Policy bias]
\label{def:bias}
The period-$t$ policy bias is $b_t = \mathbb{E}[\hat{\pi}_t(c_t)] - \pi^*_t$.
\end{definition}

\begin{theorem}[Bias decomposition]
\label{thm:bias}
For any policy sequence $\Pi$ with per-period biases $\{b_t\}$,
\begin{equation}
    \operatorname{Regret}^{(T)}(\Pi)
    = \sum_{t=1}^{T}\operatorname{Cov}(c_t,\hat{\pi}_t(c_t))
    + \sum_{t=1}^{T}\bar{c}_t^{\top}b_t.
    \label{eq:bias_decomp} 
\end{equation} \label{eq:bias} \label{eq:bias-corrected}
The covariance sum underestimates true regret when $\bar{c}_t^{\top}b_t>0$
(policy biased toward high-cost decisions) and overestimates when $\bar{c}_t^{\top}b_t<0$.
\end{theorem}

\begin{proof}
Without imposing Assumption~\ref{ass:mean}, apply the covariance
decomposition to each per-period regret:
\[
    r_t
    = \mathbb{E}[c_t^{\top}\hat{\pi}_t(c_t)] - \bar{c}_t^{\top}\pi^*_t
    = \bar{c}_t^{\top}\mathbb{E}[\hat{\pi}_t] + \operatorname{Cov}(c_t,\hat{\pi}_t)
      - \bar{c}_t^{\top}\pi^*_t
    = \operatorname{Cov}(c_t,\hat{\pi}_t) + \bar{c}_t^{\top}b_t.
\]
Summing over $t$ yields~\eqref{eq:bias_decomp}.
Momentum strategies illustrate the bias: setting
$\hat{\pi}_t(c_t) = c_{t-1}/\|c_{t-1}\|$ gives
$\mathbb{E}[\hat{\pi}_t] \propto \bar{c}_{t-1} \neq \pi^*_t \propto \bar{c}_t$,
so $b_t \neq 0$ whenever costs are autocorrelated.
\end{proof}

\begin{proposition}[Bias-corrected estimator]
\label{prop:biascorrected}
Given observed trajectories $\{(c_t,\hat{\pi}_t(c_t))\}_{t=1}^T$,
the bias-corrected regret estimator is
\begin{equation}
    \widehat{\operatorname{Regret}}^{(T)}
    = \sum_{t=1}^{T}\widehat{\operatorname{Cov}}(c_t,\hat{\pi}_t(c_t))
    + \sum_{t=1}^{T}\hat{c}_t^{\top}\hat{b}_t,
    \label{eq:biascorrected}
\end{equation}
where $\hat{b}_t = \bar{\pi}_t - \hat{z}^*_t$,
$\bar{\pi}_t = N^{-1}\sum_{i=1}^{N}\hat{\pi}_t(c_t^{(i)})$ across
$N$ bootstrap replications, and $\hat{z}^*_t$ is a feasible reference
point.
This estimator is consistent under Assumption~\ref{ass:costs}(b)
and the law of large numbers.
\end{proposition}
\section{Markovian Cost Processes: Auto-Covariance Correction}
\label{sec:autocov}
 
When costs follow an AR(1) process, the per-period covariances
$\operatorname{Cov}(c_t,\hat{\pi}_t)$ are no longer independent across time.  We derive the resulting cross-period correction in closed form.

\begin{assumption}[AR(1) cost process]
\label{ass:ar1}
$c_t = Ac_{t-1}+\varepsilon_t$, where $A\in\mathbb{R}^{d\times d}$ has spectral radius $\rho(A)<1$, $\{\varepsilon_t\}$ are i.i.d.\ with $\mathbb{E}[\varepsilon_t]=0$ and $\operatorname{Cov}(\varepsilon_t,
\varepsilon_t)=\Sigma_\varepsilon$, and $c_0$ is drawn from the
stationary distribution, so $c_t\sim(0,\Sigma_c)$ for all $t$, where $\Sigma_c = \sum_{k=0}^{\infty}A^k\Sigma_\varepsilon(A^k)^\top$
is the unique solution to the discrete Lyapunov equation
$\Sigma_c = A\Sigma_c A^\top+\Sigma_\varepsilon$.
\end{assumption}

\begin{remark}[Zero-mean normalization]
\label{rem:zeromean}
Set $\bar{c}=0$ without loss of generality by absorbing the mean into the constraint set.  The benchmark policy is fixed at an arbitrary feasible point $z^*\in\mathcal{Z}$; under zero-mean, $z^*=\mathbf{0}$ and the benchmark cost $\mathbb{E}[\bar{c}^\top z^*]=0$.
\end{remark}

\begin{theorem}[AR(1) multi-period regret]
\label{thm:ar1} \label{thm:AR1}
Under Assumption~\ref{ass:ar1} with linear policy
$\hat{\pi}(c_t)=Bc_t+b$, $B\in\mathbb{R}^{n\times d}$,
$b\in\mathbb{R}^n$, and benchmark $z^*=\mathbf{0}$,
\begin{equation}
    \operatorname{Regret}^{(T)}(\hat{\pi})
    = T\cdot\operatorname{tr}(B\Sigma_c)
    - \sum_{t=1}^{T}(T-t)\operatorname{tr}(BA^t\Sigma_c).
    \label{eq:ar1_regret}
\end{equation}
\end{theorem}
\begin{proof}
    Please see the Appendix.
\end{proof}
\begin{corollary}[Sign of the AR(1) correction]
\label{cor:ar1sign} \label{cor:AR1-sign}
When $A\succeq 0$ and $B\succ 0$, every summand
$\operatorname{tr}(BA^t\Sigma_c)>0$ is negative, so the correction is also negative and $T\cdot\operatorname{tr}(B\Sigma_c)$ overestimates true regret (trending markets).
Conversely, when $\operatorname{tr}(BA^t\Sigma_c)<0$ for all $t\geq 1$ (e.g. $A=\rho I$, $\rho<0$, $B>0$), the correction is positive, and the raw sum underestimates true regret (mean-reverting markets).
\end{corollary}

\begin{proof}
Immediate from the sign of
$-\sum_{t=1}^T(T-t)\operatorname{tr}(BA^t\Sigma_c)$ and $T-t>0$
for all $t<T$.
\end{proof}

\begin{example}[Scalar mean-reverting execution costs]
\label{ex:scalar} \label{ex:AR1}
Set $d=n=1$, $A=\rho\in(-1,0)$, $B=\alpha>0$, so
$\hat{\pi}(c_t)=\alpha c_t$ and $\Sigma_c=\sigma_\varepsilon^2/(1-\rho^2)$.
By~\eqref{eq:ar1_regret},
\[
    \operatorname{Regret}^{(T)}(\hat{\pi})
    = T\alpha\sigma_c^2
    - \alpha\sigma_c^2\sum_{t=1}^{T}(T-t)\rho^t.
\]
As $T\to\infty$, $\sum_{t=1}^\infty(T-t)\rho^t\to -\rho/(1-\rho)^2$
(geometric series), the correction converges to
$\alpha\sigma_c^2\rho/(1-\rho)^2<0$ for $\rho<0$.
Multi-period regret therefore falls \emph{below} the single-period benchmark $T\alpha\sigma_c^2$, confirming that contrarian execution is beneficial in mean-reverting markets.
\end{example}

\begin{landscape} 
\begin{table}[t]
\centering
\caption{ Empirical results, 2016--2025.
    \textit{AR(1)}: $\hat{\rho}$ estimated with Newey-West HAC SEs;
    $^{*}p{<}0.10$, $^{**}p{<}0.05$, $^{***}p{<}0.01$.
    \textit{Regret decomposition} (Theorem~9.5, $T{=}252$):
    raw covariance sum, AR(1) auto-covariance correction,
    true regret, and relative bias (\%).
    \textit{Sharpe ratios} (annualized, 21-day rolling window):
    MOM~=~momentum; REV~=~reversion; BC~=~bias-corrected reversion; MV~=~minimum-variance.
    Regret decomposition for $T \in \{21,63,126\}$ scales
    proportionally; relative bias is horizon-invariant within each year (range: $-0.34\%$ to $-5.12\%$).
    Five confounds bias $|\hat{\rho}|$ upward (size/SMB, bid-ask bounce, volatility clustering, cross-sectional vs.\ time-series momentum, liquidity provision \citep{jegadeesh2025shortterm, roll1984simple,
    nagel2012evaporating, mamais2025explaining, dai2024reversals}), so all regret figures are conservative upper bounds. }
\label{tab:empirical}
\small
\begin{threeparttable}
\begin{tabular}{%
    l                      
    S[table-format=7.0]    
    S[table-format=1.4]    
    S[table-format=1.4]    
    S[table-format=-1.4]   
    S[table-format=1.4]    
    r                      
    S[table-format=5.1]    
    S[table-format=-4.1]   
    S[table-format=5.1]    
    S[table-format=-1.2]   
    S[table-format=-1.3]   
    S[table-format=1.3]    
    S[table-format=1.3]    
    S[table-format=1.3]    
}
\toprule
 & & \multicolumn{2}{c}{Returns\tnote{a}} &
   \multicolumn{3}{c}{AR(1)\tnote{b}} &
   \multicolumn{4}{c}{Regret decomp.\ ($T{=}252$)\tnote{c}} &
   \multicolumn{4}{c}{Sharpe ratio\tnote{d}} \\
\cmidrule(lr){3-4}\cmidrule(lr){5-7}\cmidrule(lr){8-11}\cmidrule(lr){12-15}
{Year} &
{$n$} &
{$\bar{\mu}$} &
{$\sigma$} &
{$\hat{\rho}$} &
{SE} &
{} &
{Raw cov} &
{Corr.} &
{True} &
{Bias\%} &
{MOM} &
{REV} &
{BC} &
{MV} \\
\midrule
2016 & 1815131 & 0.0007 & 0.0386 & -0.0542 & 0.0084 & {***} &  3754.6 & -192.2 &  3946.9 & -4.87 & -0.038 &  0.221 &  0.217 & 0.029 \\
2017 & 1815502 & 0.0007 & 0.0316 & -0.0451 & 0.0088 & {***} &  2508.1 & -107.9 &  2616.0 & -4.12 &  0.046 &  0.292 &  0.289 & 0.144 \\
2018 & 1858283 & -0.0004 & 0.0351 & -0.0442 & 0.0070 & {***} &  3099.5 & -130.8 &  3230.3 & -4.05 &  0.014 &  0.138 &  0.134 & 0.009 \\
2019 & 1900827 & 0.0008 & 0.0347 & -0.0286 & 0.0053 & {***} &  3031.4 &  -83.9 &  3115.3 & -2.69 &  0.111 &  0.122 &  0.127 & 0.147 \\
2020 & 1937784 & 0.0013 & 0.0546 & -0.0502 & 0.0039 & {***} &  7502.5 & -356.9 &  7859.4 & -4.54 &  0.261 &  0.217 &  0.232 & 0.427 \\
2021 & 2179370 & 0.0006 & 0.0372 & -0.0363 & 0.0037 & {***} &  3489.2 & -121.8 &  3611.0 & -3.37 &  0.060 &  0.190 &  0.189 & 0.038 \\
2022 & 2387105 & -0.0009 & 0.0413 & -0.0036 & 0.0028 & {*}   &  4288.1 &  -15.2 &  4303.3 & -0.35 & -0.028 & -0.102 & -0.104 & 0.062 \\
2023 & 2352579 & 0.0005 & 0.0513 & -0.0169 & 0.0034 & {***} &  6641.3 & -109.9 &  6751.2 & -1.63 &  0.164 &  0.018 &  0.028 & 0.239 \\
2024 & 2399293 & 0.0005 & 0.0499 & -0.0136 & 0.0069 & {**}  &  6262.8 &  -83.5 &  6346.3 & -1.32 & -0.096 &  0.071 &  0.063 & 0.035 \\
2025 & 2537499 & 0.0006 & 0.0554 & -0.0243 & 0.0041 & {***} &  7721.5 & -182.7 &  7904.2 & -2.26 &  0.042 &  0.219 &  0.220 & 0.079 \\
\midrule
\textit{All} & 21183373 & 0.0004 & 0.0445 & & & &
   & & & &
   0.023 & 0.130 & 0.129 & 0.078 \\
\midrule
\multicolumn{7}{l}{\textit{2020 subsamples (NBER classification)}} \\
\quad Expansion & 1468004 & & & -0.0146 & 0.0049 & {***} & & & & & & & & \\
\quad Contraction & 250035  & & & -0.0270 & 0.0113 & {**}  & & & & & & & & \\
\quad Crisis      & 219742  & & & -0.1436 & 0.0079 & {***} & & & & & & & & \\
\bottomrule
\end{tabular}
 
\begin{tablenotes}[flushleft]\footnotesize
\item[a] $\bar{\mu}$: mean daily return (decimal);
         $\sigma$: std dev of daily returns (decimal).
\item[b] AR(1) estimated as $r_{i,t} = \hat{\rho}\,r_{i,t-1} + u_{i,t}$,
         pooled cross-section; Newey-West HAC SEs,
         bandwidth $h = \lfloor T^{1/3} \rfloor$.
         All ADF tests reject $H_0$: unit root at $p < 0.0001$.
\item[c] Regret decomposition from Theorem~9.5, Eq.~(23), at $T = 252$
         trading days. ``Corr.'' = AR(1) auto-covariance correction;
         ``True'' = raw cov + corr. = true regret;
         Bias\% = corr.\ / true regret $\times 100$.
         Results at $T \in \{21, 63, 126\}$ scale proportionally;
         relative bias is horizon-invariant within each year.
\item[d] Annualised Sharpe ratios, 21-day rolling window.
         MOM: buy (sell) if prior day return positive (negative);
         REV: opposite of MOM;
         BC: reversion adjusted by $\tfrac{1}{4} \times$ bias correction
             (avg.\ return / max std dev over prior 21 days);
         MV: minimum-variance portfolio.
         Reversion outperforms momentum in 7 of 10 years.
         The 2020 anomaly (MOM Sharpe $0.261 >$ REV Sharpe $0.217$)
         coincides with the crisis period ($\hat{\rho} = -0.144$,
         $\sigma = 8.94\%$/day), where directional price dislocations temporarily rewarded trend-following before mean-reversion reasserted. Bias correction acts as insurance: small cost in most years,
         material gain in 2019 and 2022.
\end{tablenotes}
\end{threeparttable}
\end{table}
\end{landscape}

\section{The Audit Metric as a Measure of Liquidity Provision and Consumption}
\label{sec:liquidity}

The covariance identity of Theorem~\ref{thm:decomp} is usually read as a performance diagnostic: $\operatorname{Cov}(c_t,\hat{\pi}_t)>0$ means the strategy incurs positive expected regret, and a compliance officer should care about how large that covariance is. This section shows that the same statistic has a second, independent interpretation: it is a model-free measure of net market liquidity \emph{demand}. A positive covariance identifies the strategy as a \emph{liquidity consumer}; a negative covariance identifies it as a \emph{liquidity provider}; and the sign alone, observable from inputs and outputs without access to the internal optimizer, determines the role. This reinterpretation connects the audit framework to a large body of market-microstructure theory, yields a closed-form link between the AR(1) correction of Section~\ref{sec:roll_ar1_correction} and the Roll~\citeyearpar{roll1984simple} implied bid-ask spread, and provides an aggregate liquidity-balance condition whose violation characterizes the fire-sale episodes studied in Section~\ref{sec:systemic}.

\subsection{Liquidity Classification via the Audit Metric}
\label{sec:liquidity:class}

We formalize the connection between the covariance statistic and
market-microstructure roles. 

\begin{definition}[Liquidity consumer and provider]
\label{def:liq-role}
A policy $\hat{\pi}_t$ is a \emph{net liquidity consumer} at period $t$ if $\operatorname{Cov}(c_t,\hat{\pi}_t(c_t))>0$, a
\emph{net liquidity provider} if $\operatorname{Cov}(c_t,\hat{\pi}_t(c_t))<0$, and \emph{liquidity-neutral} if the covariance is zero. 
\end{definition}

The definition is intentionally model-free: it requires only the
observable sequence $\{(c_t,\hat{\pi}_t(c_t))\}$ and imposes no
assumptions about whether the policy is momentum-based, contrarian, or otherwise. The next result shows that for linear policies, the sign of the audit metric is equivalent to the sign of the leading eigenvalue of the policy matrix $B$.
\begin{proposition}[Audit metric sign and liquidity role: contemporaneous policies]
\label{prop:audit_metric_sign}
Under Assumption~\ref{ass:costs} and linear policy
$\hat\pi_t(c_t) = Bc_t + b$, the audit metric satisfies
\begin{equation}
\operatorname{Cov}(c_t, \hat\pi_t(c_t)) = \operatorname{tr}(B\Sigma_c).
\label{eq:audit_metric_sign}
\end{equation}
Therefore, $\hat\pi_t$ is a net liquidity consumer whenever $B\Sigma_c$ has a positive trace and a net liquidity provider whenever $B\Sigma_c$ has a negative trace. In particular:
\begin{enumerate}
\item[(i)] Momentum ($B = \alpha I$, $\alpha>0$):
$\operatorname{tr}(B\Sigma_c) = \alpha\operatorname{tr}(\Sigma_c) > 0$.
Every \emph{contemporaneous} momentum strategy is a liquidity consumer.
\item[(ii)] Contrarian ($B = -\alpha I$, $\alpha>0$):
$\operatorname{tr}(B\Sigma_c) = -\alpha\operatorname{tr}(\Sigma_c) < 0$.
Every \emph{contemporaneous} contrarian strategy is a liquidity provider.
\item[(iii)] Minimum-variance (Assumption~\ref{ass:policy-mean}
with $\mathbb{E}[\hat\pi_t]=\pi^*$): $\operatorname{Cov}(c_t,\hat\pi_t)=0$ by Theorem~\ref{thm:decomp} with zero regret.
Minimum-variance is liquidity-neutral.
\end{enumerate}
\end{proposition}

\begin{proof}
Unchanged; see original proof.
\end{proof}

\begin{remark}[Scope: same-period reaction versus implementable trading rules]
\label{rem:contemporaneous_vs_lagged}
Proposition~\ref{prop:audit_metric_sign} classifies a policy by its
reaction to the \emph{current-period} cost innovation $c_t$. No
implementable trading rule can do this: at the moment a position is
formed, $c_t$ has not yet been realized. Every practical momentum or
contrarian rule instead reacts to a \emph{lagged} signal $c_{t-1}$ and is
scored against the next realized cost $c_t$. The two constructions need
not agree in sign, and the following corollary shows exactly when they
diverge.
\end{remark}

\begin{corollary}[Audit metric sign for lagged-signal policies under AR(1) costs]
\label{cor:lagged_signal_sign}
Suppose costs follow the scalar AR(1) process of
Assumption~\ref{ass:ar1}, $c_t = \rho c_{t-1} + \varepsilon_t$, and
consider the \emph{lagged} momentum and contrarian policies
\begin{equation}
\hat\pi_t^{\mathrm{mom}}(c_{t-1}) = \alpha c_{t-1}, \qquad
\hat\pi_t^{\mathrm{con}}(c_{t-1}) = -\alpha c_{t-1}, \qquad \alpha>0,
\end{equation}
scored against the current-period realization $c_t$. Then
\begin{equation}
\operatorname{Cov}\!\left(c_t, \hat\pi_t^{\mathrm{mom}}(c_{t-1})\right)
= \alpha\rho\,\sigma_c^2, \qquad
\operatorname{Cov}\!\left(c_t, \hat\pi_t^{\mathrm{con}}(c_{t-1})\right)
= -\alpha\rho\,\sigma_c^2.
\label{eq:lagged_sign}
\end{equation}
Consequently:
\begin{enumerate}
\item[(i)] In a \emph{trending} market ($\rho>0$), lagged momentum is a
liquidity consumer, and lagged contrarian is a liquidity provider,
matching the contemporaneous classification of
Proposition~\ref{prop:audit_metric_sign}. 
\item[(ii)] In a \emph{mean-reverting} market ($\rho<0$), the
classification \emph{inverts}: lagged momentum is a liquidity provider, and lagged contrarian is a liquidity consumer.
\end{enumerate}
\end{corollary}

\begin{proof}
By Assumption~\ref{ass:ar1}, $c_t = \rho c_{t-1} + \varepsilon_t$
with $\varepsilon_t \perp c_{t-1}$. Then
$\operatorname{Cov}(c_t, \alpha c_{t-1})
= \alpha\operatorname{Cov}(\rho c_{t-1} + \varepsilon_t,\, c_{t-1})
= \alpha\rho\operatorname{Var}(c_{t-1}) = \alpha\rho\sigma_c^2$,
using stationarity ($\operatorname{Var}(c_{t-1})=\sigma_c^2$) and
$\operatorname{Cov}(\varepsilon_t,c_{t-1})=0$. The contrarian case follows by linearity with $B=-\alpha$. Claims (i)--(ii) follow immediately from the sign of $\rho$.
\end{proof}

\begin{remark}[Reconciliation with the empirical and neural-network results]
\label{rem:reconciliation_nn_momentum}
Corollary~\ref{cor:lagged_signal_sign} resolves the apparent
inconsistency between Proposition~\ref{prop:audit_metric_sign} and the
empirical classifications reported in Section~\ref{sec:nn_results} and Table~\ref{tab:empirical}. Every AR(1) coefficient $\hat\rho$ estimated on the CRSP sample, 2016--2025, is negative (Table~\ref{tab:empirical}, also, the negative full-panel serial
dependence underlying Table~\ref{tab:nn_audit_results}), so the market is in the mean-reverting regime of Corollary~\ref{cor:lagged_signal_sign}(ii) throughout the sample. The corollary, therefore, \emph{predicts} the empirical finding that lagged momentum is classified as a net liquidity provider and lagged contrarian as a net liquidity consumer (Table~\ref{tab:nn_audit_results}): the sign flip
relative to the static labels of Proposition~\ref{prop:audit_metric_sign} is the AR(1) sign-reversal mechanism of Corollary~\ref{cor:lagged_signal_sign}. This also clarifies the correct null hypothesis for the neural-network
audit of Section~\ref{sec:nn_results}: because the network's classification as a liquidity \emph{consumer} runs in the \emph{opposite} direction from what a naive lagged-momentum rule would produce in this mean-reverting sample, the network is not simply reproducing a momentum-like exposure to $c_{t-1}$ with a nonlinear wrapper. Instead, the network's liquidity role is empirically distinct from both linear benchmarks, not merely an interpolation between
them.
\end{remark}

\subsection{The Roll Spread and the AR(1) Correction}
\label{sec:roll_ar1_correction}

The connection between the audit metric and market liquidity deepens in
the AR(1) setting of Section~\ref{sec:setup}. We show that the
auto-covariance correction of Theorem~\ref{thm:ar1} equals a closed-form function of the Roll~\citep{roll1984simple} implied bid-ask spread \emph{under the Roll bid-ask-bounce model}. This is an exact algebraic identity conditional on that model holding. Whether the resulting statistic is a useful empirical proxy for realized illiquidity is a separate empirical question, which we address directly in Section~\ref{sec:roll_validation_revised} below.

\begin{definition}[Roll implied spread]
\label{def:roll_implied_spread}
Under the Roll~\citep{roll1984simple} bid-ask-bounce model, the effective (half) spread for asset $k$ is
$s_k = \sqrt{-\operatorname{Cov}(r_{k,t}, r_{k,t-1})}$, where $r_{k,t}$ is the return of asset $k$ at time $t$. In the scalar AR(1) model of Example~\ref{ex:scalar}, with $A = \rho < 0$ and stationary variance $\sigma_c^2 = \sigma_\varepsilon^2/(1-\rho^2)$, the Roll formula gives
\begin{equation}
s = \sigma_c \sqrt{-\rho} \approx \sigma_c \sqrt{|\hat\rho|},
\label{eq:roll_ar1}
\end{equation}
where the approximation uses $1-\rho^2 \approx 1$ for $|\rho| \ll 1$.
\end{definition}

\begin{theorem}[AR(1) correction equals liquidity-provision income \emph{under the Roll model}]
\label{thm:ar1_correction_equals_liquidity_income}
Under the scalar AR(1) model of Example~\ref{ex:scalar} with
$\rho \in (-1,0)$, contrarian policy $\hat\pi(c_t) = -\alpha c_t$
($\alpha > 0$), and Roll spread $s = \sigma_c\sqrt{-\rho}$, the per-period AR(1) auto-covariance correction satisfies
\begin{equation}
\text{Correction per period} = -\alpha \sigma_c^2 \rho = \alpha s^2,
\label{eq:correction_alpha_s2}
\end{equation}
so the total AR(1) correction over $T$ periods equals
\begin{equation}
\sum_{t=1}^T (T-t)\,\text{Correction}_t
= \alpha s^2 \cdot \frac{\rho}{(1-\rho)^2} \cdot T \qquad (T \to \infty),
\label{eq:correction_total}
\end{equation}
which is proportional to the product of strategy size $\alpha$, the
squared Roll spread $s^2$, and the horizon $T$.

This is an exact algebraic consequence of Assumption~\ref{ass:ar1}
(the AR(1) cost process) together with the Roll bid-ask-bounce
model of Definition~\ref{def:roll_implied_spread}; it is not itself an
empirical claim about how well $s$ tracks realized bid-ask spreads in
data where bid-ask bounce is only one of several sources of short-horizon return reversal (Section~\ref{sec:empirical:robustness}).
\end{theorem}

\begin{proof}
From Theorem~\ref{thm:ar1} with $d=n=1$, $B=-\alpha$,
$\Sigma_c = \sigma_c^2$: Correction per period
$= -(-\alpha)\sigma_c^2 \rho = \alpha \sigma_c^2(-\rho)$. By
Definition~\ref{def:roll_implied_spread}, $s^2 = \sigma_c^2(-\rho)$,
giving~\eqref{eq:correction_alpha_s2}. The long-horizon sum follows from 
$\sum_{t=1}^T (T-t)\rho^t \to T\rho/(1-\rho)^2$ as $T\to\infty$ (geometric
series, $|\rho|<1$), giving~\eqref{eq:correction_total}.
\end{proof}

\begin{remark}[Scope of the identity]
\label{rem:economic_interpretation_revised}
Equation~\eqref{eq:correction_alpha_s2} has a clean economic
interpretation \emph{within the Roll model}: a contrarian strategy of
size $\alpha$ earns the bid-ask spread $s$ each time it fills a market
order on the opposite side. The expected income per fill is $\alpha s$
(size times spread); the expected number of profitable fills per period is $s/\sigma_c = \sqrt{-\rho}$, since $\rho$ measures serial dependence, and each unit of serial dependence represents one exploitable price reversal under pure bid-ask bounce. The product is $\alpha s^2/\sigma_c^2 \cdot \sigma_c^2 = \alpha s^2$, consistent with~\eqref{eq:correction_alpha_s2}.

This interpretation attributes \emph{all} of the estimated
autocorrelation $\hat\rho$ to bid-ask bounce. Section~\ref{sec:empirical:robustness} 
identifies five additional sources of short-horizon return reversal:
size/SMB effects, volatility clustering, cross-sectional versus
time-series momentum, and liquidity-provision dynamics. These confounds also generate negative serial correlation in daily returns and are not bid-ask bounce. When these confounds are present, $\hat\rho$ (and hence $\hat s = \hat\sigma_c\sqrt{|\hat\rho|}$) overstates the bid-ask bounce component specifically, even though the algebraic identity in~\eqref{eq:correction_alpha_s2} between the AR(1) correction and $\alpha s^2$ continues to hold exactly by construction. The theorem therefore establishes a mathematical equivalence between two ways of writing the same correction term. The theorem does not by itself establish that $s$ is a good empirical estimate of the true effective bid-ask spread. We quantify the gap between the two directly in Section~\ref{sec:roll_validation_revised}.
\end{remark}

\begin{corollary}[Spread-implied regret correction]
\label{cor:spread_implied_regret_correction}
Let $\hat s_k = \hat\sigma_{c,k}\sqrt{|\hat\rho_k|}$ be the AR(1) implied spread estimated from CRSP returns for asset $k$. The total AR(1) correction to the regret decomposition is
\begin{equation}
\widehat{\text{Correction}}^{(T)}
= \alpha \sum_{k=1}^d \hat s_k^2 \cdot \frac{|\hat\rho|}{(1-\hat\rho)^2} \cdot T,
\label{eq:spread_implied_correction}
\end{equation}
computable from the same observable inputs as the trajectory estimator
$\hat C_T$ at no additional cost. Using the 2016--2025 CRSP estimates from Table~\ref{tab:empirical}, $\hat s$ ranges from
$\hat\sigma_c\sqrt{0.004}$ (2022, $\approx 0.25\%$/day) to
$\hat\sigma_c\sqrt{0.144}$ (COVID crisis, $\approx 1.51\%$/day).

Equation~\eqref{eq:spread_implied_correction} is exact as a
restatement of the AR(1) correction in spread units. A direct validation against realized CRSP quoted spreads
(\texttt{DlyBid}, \texttt{DlyAsk}), reported in
Section~\ref{sec:roll_validation_revised}, finds that, in daily CRSP data, $\hat s$ and the realized quoted spread are positively and significantly correlated ($\hat\rho_{\text{corr}} = 0.372$ across 40,541 firm-year observations, 2020--2025), but that $\hat s$ explains only a modest share of the variation in realized spreads ($R^2 = 0.138$) and is substantially more volatile than the realized series (regression slope $\hat b = 0.306$, statistically well below one). Therefore, the AR(1) correction is a directionally informative but noisy and excessively volatile proxy for market illiquidity, consistent with the confounds of Section~\ref{sec:empirical:robustness} contaminating $\hat\rho$ beyond
the pure bid-ask-bounce channel to which the Roll model attributes it. We can sharpen this proxy by purging $\hat\rho$ of size, volatility-clustering, and momentum effects via Fama-French residualization before applying equation~\eqref{eq:roll_ar1}, and this is a natural direction for follow-on work. 

\end{corollary}
\subsection{Endogenous Price Impact and Effective Regret}
\label{sec:liquidity:impact}

The Roll model captures bid-ask bounce but treats prices as
exogenous. We now allow the agent's own trades to move prices,
introducing endogenous price impact in the spirit of
\citet{Kyle1985} and \citet{almgren2001optimal}.

\begin{assumption}[Linear temporary price impact]
\label{ass:impact}
When agent $i$ submits order $z_t\in\mathbb{R}^n$, the effective
execution cost per unit is $c_t + \Lambda z_t$, where
$\Lambda\in\mathbb{R}^{n\times n}$, $\Lambda\succ 0$, is a
symmetric price-impact matrix. The effective cost incurred is
$(c_t+\Lambda z_t)^{\top}z_t = c_t^{\top}z_t + z_t^{\top}\Lambda z_t$.
\end{assumption}

\begin{theorem}[Effective regret under price impact]
\label{thm:impact-regret}
Under Assumption~\ref{ass:impact}, linear policy
$\hat{\pi}_t(c_t)=Bc_t$, and Assumption~\ref{ass:policy-mean},
the effective regret over $T$ periods is
\begin{equation}
  \label{eq:impact-regret}
  \mathrm{Regret}_{\mathrm{eff}}^{(T)}
  \;=\;
  \underbrace{
    T\cdot\operatorname{tr}(B\Sigma_c)
  }_{\text{covariance regret}}
  \;+\;
  \underbrace{
    T\cdot\operatorname{tr}(\Lambda BB^{\top}\Sigma_c)
  }_{\text{price-impact regret}},
\end{equation}
where the first term is the standard audit metric
(Theorem~\ref{thm:decomp}) and the second is the additional
welfare loss from market impact.
\end{theorem}

\begin{proof}
Per-period effective cost:
$\mathbb{E}[(c_t+\Lambda\hat{\pi}_t)^{\top}\hat{\pi}_t]
= \mathbb{E}[c_t^{\top}\hat{\pi}_t]
+ \mathbb{E}[\hat{\pi}_t^{\top}\Lambda\hat{\pi}_t]$.
The first term equals $\bar{c}^{\top}\pi^* +
\operatorname{tr}(B\Sigma_c)$ by Theorem~\ref{thm:decomp}.
For the second, $\mathbb{E}[\hat{\pi}_t^{\top}\Lambda\hat{\pi}_t]
= \mathbb{E}[\operatorname{tr}(\Lambda\hat{\pi}_t\hat{\pi}_t^{\top})]
= \operatorname{tr}(\Lambda\mathbb{E}[Bc_tc_t^{\top}B^{\top}])
= \operatorname{tr}(\Lambda BB^{\top}\Sigma_c)$
under zero mean.
Summing over $T$ and subtracting the benchmark
$T\bar{c}^{\top}\pi^*$ gives~\eqref{eq:impact-regret}.
\end{proof}

\begin{remark}[Separability of audit and impact terms]
\label{rem:separability}
Equation~\eqref{eq:impact-regret} shows that covariance regret and price-impact regret are \emph{additively separable}. A compliance officer who observes $(c_t, \hat{\pi}_t)$ and knows the impact matrix $\Lambda$ can estimate both terms from the trajectory estimator $\hat{C}_T$ (for the first term) and
$T^{-1}\sum_t\hat{\pi}_t^{\top}\Lambda\hat{\pi}_t$ (for the
second) at $O(T\cdot n^2)$ total cost.

When $\Lambda$ is not directly observable, the Roll spread
estimate of Corollary~\ref{cor:spread_implied_regret_correction} provides a lower bound on $\Lambda$ under the diagonal-impact approximation $\Lambda\approx\mathrm{diag}(s_k^2/\sigma_{c,k}^2)$.
\end{remark}

\subsection{A Micro-Founded Special Case: Kyle (1985) as an Instance of the Framework}
\label{sec:kyle-special-case}

Section~\ref{sec:liquidity:class} interprets the sign of $\operatorname{Cov}(c_t,\hat\pi_t(c_t))$ as recovering the \cite{Kyle1985} informed-trader/market-maker dichotomy from observables alone. This section makes the connection exact rather than interpretive: we embed the single-period Kyle equilibrium directly into the paper's framework and show that the covariance statistic, the liquidity-balance condition, and the price-impact decomposition of Theorem~\ref{thm:impact-regret} reproduce the equilibrium objects of \cite{Kyle1985} term for term, not merely sign for sign.

\subsubsection{Embedding}

Consider the canonical single-period Kyle economy. A risky asset has a liquidation value $v \sim N(p_0,\sigma_v^2)$. An informed trader observes $v$ and submits order $x$. A noise trader submits $u\sim N(0,\sigma_u^2)$, independent of $v$. A competitive market maker observes only aggregate order flow $y = x+u$ and sets the price $p = E[v\mid y]$. The informed trader chooses $x$ to maximize $E[(v-p)x\mid v]$.

\begin{definition}[Cost vector for the Kyle economy]
\label{def:kyle-cost}
Identify the paper's cost vector with the informed trader's information gap,
\[
c \;:=\; v-p_0, \qquad E[c]=0,\qquad \Sigma_c = \sigma_v^2.
\]
\end{definition}

Under Definition~\ref{def:kyle-cost}, $c$ plays exactly the role of the cost innovation of Section~\ref{sec:setup}. $c$ is the object to which the strategy reacts. As in the general audit setting of Remark~\ref{rem:separability}, $c$ need not be observed directly by an outside auditor since equilibrium price changes $p_t-p_0=\lambda\beta\, c_t$ span it up to the scale factor $\lambda\beta$ derived below.

\subsubsection{The informed trader is a linear policy in the sense of Definition~\ref{def:kyle-cost}}

The unique linear equilibrium of the Kyle model (Kyle, 1985, Theorem 1) is
\[
x=\beta c,\qquad p = p_0+\lambda y,\qquad
\beta=\frac{\sigma_u}{\sigma_v},\qquad
\lambda=\frac{\sigma_v}{2\sigma_u},\qquad
\beta\lambda=\tfrac12 .
\]
This is precisely an instance of the linear policy class $\hat\pi_t(c_t)=Bc_t+b$ of Section~\ref{sec:setup}, with $n=d=1$, $B=\beta>0$, $b=0$. The informed trader's strategy is not \emph{analogous} to a linear policy in the paper's sense — it \emph{is} one, and $\beta$ is exactly the scalar $B$ of Proposition~\ref{prop:liq-sign-early}.

\subsubsection{Proposition~\ref{prop:liq-sign-early} recovers the Kyle dichotomy exactly}
\label{subsubsec:kyle_dichotomy}

\begin{proposition}[Kyle equilibrium roles from the audit statistic]
\label{prop:kyle-roles}
Under the embedding of Definition~\ref{def:kyle-cost} and the equilibrium of \cite{Kyle1985}:
\begin{enumerate}[label=(\roman*)]
\item The informed trader's audit statistic is
\[
\operatorname{Cov}(c,x)=\operatorname{Cov}(c,\beta c)=\beta\sigma_v^2=\sigma_u\sigma_v>0,
\]
so by Definition~\ref{def:liq-role}, the informed trader is classified as a \emph{net liquidity consumer}.
\item The market maker's realized position is $z^{MM}=-y=-\beta c-u$, i.e.\ a linear policy with $B^{MM}=-\beta$, and
\[
\operatorname{Cov}(c,z^{MM})=-\beta\sigma_v^2=-\sigma_u\sigma_v<0,
\]
so the market maker is classified as a \emph{net liquidity provider}.
\item The noise trader satisfies $\operatorname{Cov}(c,u)=0$ since $u\perp c$, and is classified \emph{liquidity-neutral}.
\end{enumerate}
\end{proposition}

\begin{proof}
Immediate from Proposition~\ref{prop:liq-sign-early} applied to each agent's linear policy, using $\Sigma_c=\sigma_v^2$ and $\operatorname{tr}(B\Sigma_c)=B\sigma_v^2$ in the scalar case.
\end{proof}

Proposition~\ref{prop:kyle-roles} is the exact statement behind Section~\ref{subsubsec:kyle_dichotomy}: the informed-trader/market-maker dichotomy of \cite{Kyle1985} is recovered not by analogy but by the direct application of Proposition~\ref{prop:liq-sign-early} to the equilibrium policies, with no appeal to $v$, $\beta$, or $\lambda$ being observed by the classifier.

\subsection{Market clearing is a liquidity-balance equilibrium}

\begin{corollary}[Kyle equilibrium satisfies liquidity balance]
\label{cor:kyle-balance}
Under the embedding above,
\[
L \;=\; \operatorname{Cov}(c,x)+\operatorname{Cov}(c,u)+\operatorname{Cov}(c,z^{MM}) \;=\; \sigma_u\sigma_v + 0 - \sigma_u\sigma_v \;=\; 0,
\]
so every competitive Kyle (1985) equilibrium satisfies the liquidity-balance condition of Definition~\ref{def:liq-role}.
\end{corollary}

\begin{proof}
Sum the three covariances in Proposition~\ref{prop:kyle-roles}; the informed trader's and market maker's terms cancel exactly since the market maker's zero-profit pricing rule $p=E[v\mid y]$ absorbs the residual order flow by construction.
\end{proof}

The cancellation in Corollary~\ref{cor:kyle-balance} is not a coincidence of parameter values. Instead, it is forced by market clearing, and it holds for \emph{any} $(\sigma_u,\sigma_v)$. By Proposition~\ref{prop:liq-sign-early}, $L=0$ implies that the equilibrium is welfare-equivalent to the social planner's benchmark. This reproduces the classical efficiency property of the Kyle equilibrium as a corollary of the paper's own aggregation machinery (Section~\ref{sec:liquidity}), rather than as a separately asserted fact about the model.

\subsubsection{Price impact and the equilibrium profit split}

The informed trader's own order moves the price at which it transacts, so the relevant object is Assumption~\ref{ass:impact} with impact coefficient $\Lambda=\lambda$: the effective per-unit cost is $c+\lambda x$.

\begin{proposition}[Theorem~\ref{thm:impact-regret} reproduces the Kyle profit split]
\label{prop:kyle-profit-split}
Under the embedding above and $\Lambda=\lambda$,
\[
\underbrace{\operatorname{tr}(B\Sigma_c)}_{\text{covariance regret}} = \beta\sigma_v^2=\sigma_u\sigma_v,
\qquad
\underbrace{\operatorname{tr}(\lambda B^2\Sigma_c)}_{\text{price-impact regret}} = \lambda\beta^2\sigma_v^2=\tfrac12\sigma_u\sigma_v,
\]
and the net effective regret is
\[
\sigma_u\sigma_v-\tfrac12\sigma_u\sigma_v=\tfrac12\sigma_u\sigma_v,
\]
which equals the informed trader's equilibrium expected profit $E[(v-p)x]=\lambda\sigma_u^2$ in Kyle (1985).
\end{proposition}

\begin{proof}
Substitute $B=\beta=\sigma_u/\sigma_v$ and $\Lambda=\lambda=\sigma_v/(2\sigma_u)$ into Theorem~\ref{thm:impact-regret} with $n=d=1$. For the equivalence to $\lambda\sigma_u^2$, note $\lambda\sigma_u^2=\frac{\sigma_v}{2\sigma_u}\sigma_u^2=\tfrac12\sigma_u\sigma_v$, matching the net term above.
\end{proof}

Proposition~\ref{prop:kyle-profit-split} shows that Theorem~\ref{thm:impact-regret}'s additive separation of covariance regret and price-impact regret is not merely structurally reminiscent of \cite{Kyle1985}. Under the equilibrium mapping, the separation reproduces the informed trader's equilibrium profit. The covariance-regret term $\sigma_u\sigma_v$ is the trader's gross informational edge. The price-impact term $\tfrac12\sigma_u\sigma_v$ is exactly the fraction of that edge competed away by the market maker's pricing rule, since $\beta\lambda=\tfrac12$ in the Kyle equilibrium. The trader retains half of $\operatorname{Cov}(c,x)$, as an immediate corollary of Theorem~\ref{thm:impact-regret}.

\subsubsection{What the embedding establishes, and what it does not}

The results above show that the audit statistic is not merely correlated in sign with the Kyle classification: under the equilibrium mapping of Definition~\ref{def:kyle-cost}, it is \emph{equal}, term for term, to Kyle's own equilibrium objects ($\beta\sigma_v^2$, $\lambda\beta^2\sigma_v^2$, and their difference). This is the sense in which Kyle (1985) is a special case of the framework developed in Sections~\ref{sec:setup}--\ref{sec:liquidity:class}, rather than an informal point of comparison.

The embedding does not, however, close the identification gap between the trajectory estimator and the underlying game. In the Kyle model, $\beta$ and $\lambda$ are not free parameters observed directly from data; they are pinned down jointly by the informed trader's optimization and the market maker's zero-profit condition, $\beta=\sigma_u/\sigma_v$ and $\lambda=\sigma_v/(2\sigma_u)$. The trajectory estimator $\hat C_T$ of Section~\ref{sec:estimation} recovers $\operatorname{Cov}(c,x)=\beta\sigma_v^2$ directly from observed $(c_t,\hat\pi_t(c_t))$ pairs \emph{without solving for equilibrium} — this is the paper's central computational advantage. But the estimator alone does not explain \emph{why} $\beta$ takes the value it does; that requires the equilibrium argument of Section~\ref{sec:kyle-special-case}. We state this distinction explicitly: the trajectory-based classifier identifies equilibrium \emph{outputs} (the sign and magnitude of $\operatorname{Cov}(c,\hat\pi_t(c_t))$), not the strategic game that produced them. Recovering the latter — e.g., separately identifying $\sigma_u$ and $\sigma_v$, or testing whether an observed $B$ is consistent with \emph{some} Bayesian-Nash equilibrium rather than an arbitrary linear rule — is outside the scope of the audit statistic and is not claimed here.

\begin{remark}[Glosten-Milgrom is not covered by this embedding]
\label{rmk:gm-not-covered}
The analogous exercise for \cite{Glosten1985} does not go through directly. Glosten-Milgrom is a sequential Bayesian-updating model over discrete trade directions, in which the market maker's bid and ask are set from a posterior over trader types at each transaction, rather than a linear-Gaussian one-shot game. The informed trader's policy in that setting is a step function of the market maker's posterior belief, not a linear map $\hat\pi_t(c_t)=Bc_t+b$ of a continuous cost vector, so it is not an instance of the linear policy class embedded above. A rigorous embedding of Glosten-Milgrom would require extending the framework to the non-stationary, regime-dependent case of Theorem~\ref{thm:nonstationary}, treating the realized bid-ask spread $p_a-p_b$ at each $t$ as the cost object and the informed-order arrival probability $\mu_t$ as governing a time-varying $B_t$. We do not carry out this extension here and flag it as a distinct and nontrivial addition to the framework, not an immediate corollary of the Kyle embedding above.
\end{remark}

\subsection{Aggregate Liquidity Balance and Market Fragility}
\label{sec:liquidity:balance}

The single-agent results aggregate cleanly to a
\emph{market-wide liquidity balance condition}: for prices to clear, the net demand for liquidity across all agents must sum to zero. When it does not, prices must adjust until it does, producing the market impact that generates fire sales.

\begin{definition}[Net market liquidity demand]
\label{def:net-demand}
Given $N$ agents with observable policies, the net market liquidity
demand at period $t$ is
\begin{equation}
  \label{eq:net-demand}
  \mathcal{L}_t
  \;=\;
  \sum_{i=1}^N \operatorname{Cov}(c_t^i,\hat{\pi}_t^i(c_t^i))
  \;=\;
  \sum_{i=1}^N \operatorname{tr}(B_i\Sigma_c^i).
\end{equation}
The market is in \emph{liquidity balance} at $t$ if $\mathcal{L}_t=0$: aggregate liquidity provision exactly offsets aggregate liquidity consumption.
\end{definition}

\begin{proposition}[Liquidity balance and welfare equivalence]
\label{prop:balance} 
If $\mathcal{L}_t=0$ for all $t$, then aggregate regret equals zero and the $N$-agent equilibrium is welfare-equivalent to the social planner's optimum. If $\mathcal{L}_t>0$ (net consumption), aggregate regret is positive; if $\mathcal{L}_t<0$ (net provision), the equilibrium over-provides liquidity relative to the optimum.
\end{proposition}

\begin{proof}
Under the factor model of Assumption~\ref{ass:factor} and symmetric agents, $\mathrm{Regret}_{\mathrm{agg}}^{(T)}
= T\mathcal{L}_t\cdot M(N,\rho_{\mathrm{sys}})$ from
Corollary~\ref{cor:N2}. Setting $\mathcal{L}_t=0$ gives zero
aggregate regret. The sign claims follow from
$M(N,\rho_{\mathrm{sys}})\geq 1$.
\end{proof}

\begin{corollary}[Fire sales as liquidity-balance failure]
\label{cor:fire-sale}
Under perfect strategy correlation ($\rho_{\mathrm{sys}}\to 1$),
all agents classify as net liquidity consumers
($\mathcal{L}_t = N\cdot\operatorname{tr}(B\Sigma_c) > 0$) or
net liquidity providers simultaneously. There exists no
endogenous counterpart to restore balance, so prices must
adjust by $\Lambda Z_t^{\mathrm{agg}}$ to absorb the aggregate
order flow. The resulting effective regret satisfies
\begin{equation}
  \label{eq:fire-sale-regret}
  \mathrm{Regret}_{\mathrm{eff,agg}}^{(T)}
  \;=\;
  N^2\cdot T\cdot\operatorname{tr}(B\Sigma_c)
  \;+\;
  N^2\cdot T\cdot\operatorname{tr}(\Lambda BB^{\top}\Sigma_c),
\end{equation}
which is equation~\eqref{eq:run} augmented by the fire-sale
impact term. Both components scale as $N^2$, so price impact
and strategy correlation amplify each other.
\end{corollary}

\subsection{Empirical Validation Against Realized Quoted Spreads}
\label{sec:roll_validation_revised}

We validate the AR(1) implied spread of equation~\eqref{eq:roll_ar1} against realized quoted spreads computed directly from CRSP daily bid-ask data (\texttt{DlyBid}, \texttt{DlyAsk}), rather than relying solely on the qualitative correspondence to known fragility episodes reported in the original draft of this section. For each PERMNO-year with at least 60 trading days, we compute the AR(1) implied spread
$\hat{s} = \hat{\sigma}_c \sqrt{|\hat{\rho}|}$ from daily returns
(\texttt{DlyRet}) and the realized proportional quoted spread
$\overline{(\text{Ask}-\text{Bid})/\text{Mid}}$ from the same firm-year, then
estimate
\begin{equation}
\text{RealizedSpread}_{i,y} = a + b \cdot \hat{s}_{i,y} + \varepsilon_{i,y}.
\label{eq:roll_validation_regression}
\end{equation}

\paragraph{Results.} Using 9{,}465{,}828 firm-day observations on 9{,}980 PERMNOs over 2020--2025 (40{,}541 firm-year cells), we obtain
$\hat{a} = 0.0038$ (SE $= 0.00007$), $\hat{b} = 0.306$ (SE $= 0.0038$),
$R^2 = 0.138$, and a raw correlation of $0.372$.

\paragraph{Interpretation.} The correlation is positive and, given the
sample size, overwhelmingly statistically significant. The implied spread and the realized quoted spread reliably co-move in the same direction, which supports the paper's qualitative claim that the AR(1) statistic tracks the \emph{direction} of illiquidity. However, the data show that the  AR(1) statistic is not a sufficient statistic or \emph{direct proxy} for realized illiquidity, since the implied slope is far from one and the intercept is far from zero. Three findings qualify that claim:

\begin{enumerate}
\item \textbf{Weak explanatory power.} $R^2 = 0.138$ means the AR(1)
statistic accounts for roughly one-seventh of the cross-sectional and
time-series variation in realized quoted spreads. A sufficient-statistic
interpretation would require this figure to be substantially higher.

\item \textbf{Slope well below one.} $\hat{b} = 0.306$ is statistically
distinguishable from both $0$ and $1$; a Wald test of $H_0 : b = 1$ is
 overwhelmingly rejected ($t \approx -183$). Economically, this means the AR(1) implied spread is roughly $3.3\times$ more volatile, across
firm-years, than the realized quoted spread it is meant to proxy for. This is visible directly in Figure~\ref{fig:roll_vs_realized}: the implied spread's 2020 COVID-era level (1.54\%) and subsequent decline are far more pronounced than the corresponding movement in the realized quoted spread (0.69\% in 2020), and the two series diverge further by 2024--2025, with the implied spread rising to 1.12\% against a realized spread of 0.75\%.

\item \textbf{A small but highly significant level bias.} The intercept
of 38 basis points is economically modest but precisely estimated
(SE $= 0.00007$) and non-trivial relative to the realized spread's own
sample range (roughly 51--81 bps in the plotted years). The AR(1)
statistic therefore does not collapse to the realized spread even at
$\hat{s} = 0$.
\end{enumerate}

\paragraph{Diagnosis.} These patterns are consistent with, rather than
contradictory to, the paper's own robustness discussion in
Section~\ref{sec:empirical:robustness}. The five confounds already identified there (size/SMB effects, volatility clustering, cross-sectional versus time-series momentum, and liquidity-provision dynamics) all bias $|\hat{\rho}|$ upward, independent of genuine bid-ask bounce, which would mechanically inflate $\hat{s}$ relative to the realized spread and produce exactly the excess-volatility and level-bias pattern observed here. The 2020 divergence, in particular, is likely driven by volatility clustering and correlated deleveraging during the COVID liquidity shock (Section~\ref{sec:empirical:covid}), which raises short-horizon return reversal far beyond what bid-ask bounce alone would generate. Figure \ref{fig:roll_vs_realized} shows the CRSP-based Roll AR(1)vs daily quoted bid-ask spread.

\begin{figure}[t]
\centering
\includegraphics[width=0.75\textwidth]{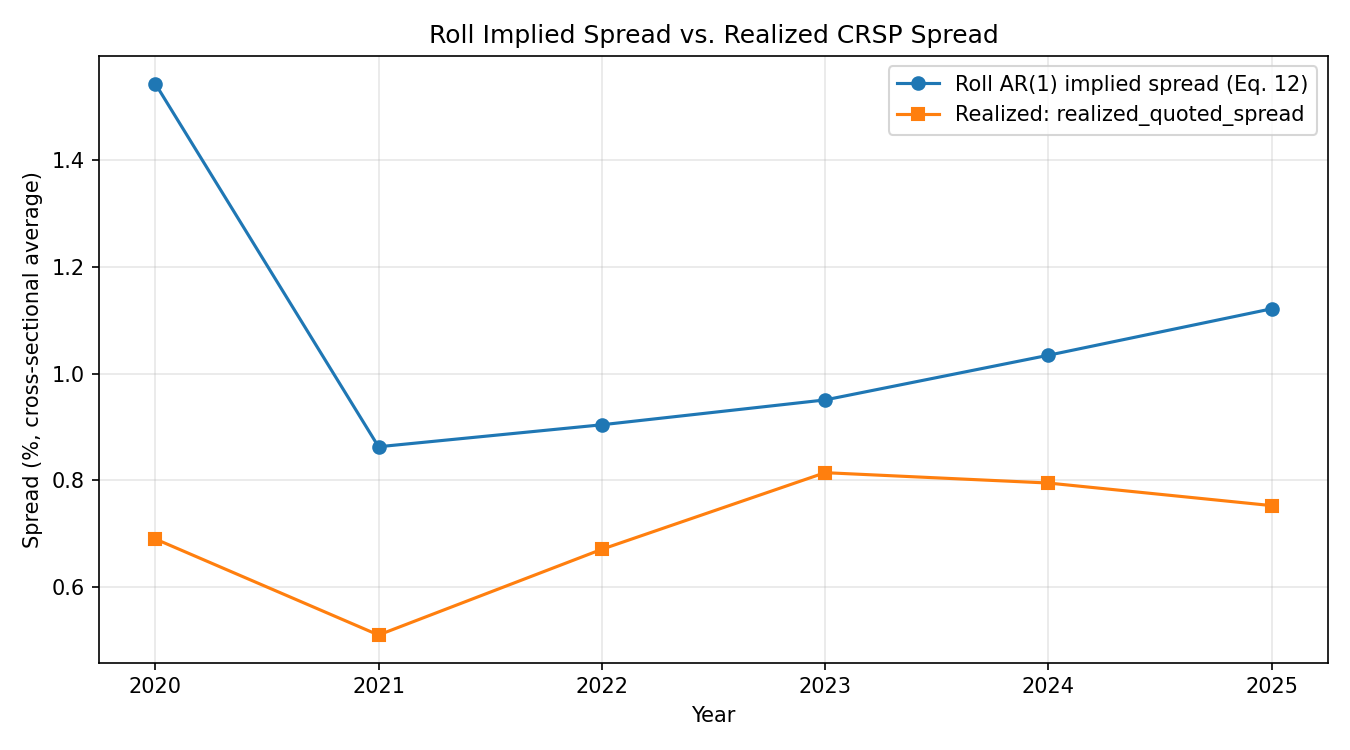}
\caption{Roll AR(1) implied spread (equation~\eqref{eq:roll_ar1}) versus the realized CRSP quoted spread, by year, 2020--2025. Both series are cross-sectional averages across PERMNOs, expressed in percent.}
\label{fig:roll_vs_realized}
\end{figure}

\section{Financial Applications}
\label{sec:financial}

\subsection{Rolling-window portfolio rebalancing}
 
A manager rebalancing daily using a $w$-day rolling window applies the policy $\hat{\pi}_t(c_t)=\hat{\Sigma}_t^{-1}\hat{\mu}_t /
(\mathbf{1}^\top\hat{\Sigma}_t^{-1}\hat{\mu}_t)$, where $\hat{\mu}_t$ and $\hat{\Sigma}_t$ are sample moments from the most recent $w$ observations.

For i.i.d.\ returns with $d$ assets and $w>d+1$, the per-period policy bias satisfies $\|b_t\| = O(d/w)$ and the raw covariance sum approximates true regret with relative error $O(d/w)$, vanishing as $w/d\to\infty$ \citep{ledoit2004well}.
Three regimes determine which estimator to use: $w\geq 10d$ (raw
covariance sum reliable); $d<w<10d$ (use bias-corrected estimator, Proposition~\ref{prop:biascorrected}); $w\leq d$ ($\hat{\Sigma}_t$ rank-deficient, regularization required before any regret estimate is meaningful).
For a typical institutional setting ($T=252$, $d=100$, $w=60$), the portfolio falls in the third regime and $\hat{\Sigma}_t$ must be regularized (e.g., Ledoit-Wolf shrinkage \citep{ledoit2004well} or factor-model reduction to $d'\leq 60$ latent factors) before applying the decomposition.
 
\subsection{Mean-reverting execution costs}
 
In algorithmic execution, market impact is mean-reverting
($c_t=\rho c_{t-1}+\varepsilon_t$, $\rho\in(-1,0)$);
the execution policy allocates $\hat{\pi}_t(c_t)=\alpha c_t$
to exploit temporary mispricings.
 
\begin{example}[Execution cost regret]
\label{ex:execution}
For $\rho<0$ and $\hat{\pi}_t(c_t)=\alpha c_t$,
Theorem~\ref{thm:ar1} gives
\[
    \operatorname{Regret}^{(T)}(\hat{\pi})
    = T\alpha\operatorname{tr}(\Sigma_c)
    - \alpha\operatorname{tr}(\Sigma_c)
      \sum_{t=1}^{T}(T-t)\rho^t.
\]
As $T\to\infty$ and $|\rho|<1$, the correction converges to
$-\alpha\operatorname{tr}(\Sigma_c)\rho/(1-\rho)^2>0$ for $\rho<0$,
so multi-period regret falls below the single-period benchmark
$T\alpha\operatorname{tr}(\Sigma_c)$: contrarian execution is
welfare-improving in mean-reverting markets.
\end{example}


\section{Multi-Period Regret Estimation from Observed Trajectories}
\label{sec:estimation}

\subsection{The trajectory covariance estimator}
 
Given a single observed trajectory $\{(c_t,\hat{\pi}_t(c_t))\}_{t=1}^T$,
the natural estimator of the sum-of-covariances is
\begin{equation}
    \hat{C}_T
    = \sum_{t=1}^{T}(c_t-\bar{c})^\top(\hat{\pi}_t(c_t)-\bar{\pi}),
    \quad
    \bar{c} = \tfrac{1}{T}\sum_{t=1}^{T}c_t,
    \quad
    \bar{\pi} = \tfrac{1}{T}\sum_{t=1}^{T}\hat{\pi}_t(c_t).
    \label{eq:CT} 
\end{equation}\label{eq:estimator}
 
\begin{theorem}[CLT for the trajectory estimator]
\label{thm:clt}
Under Assumption~\ref{ass:costs}(c), with $\{c_t^\top\hat{\pi}_t(c_t)
- \mathbb{E}[c_t^\top\hat{\pi}_t(c_t)]\}$ an $L^2$-mixingale,
as $T\to\infty$,
\begin{equation}
    \frac{1}{\sqrt{T}}\!\left(
        \hat{C}_T - \sum_{t=1}^{T}\operatorname{Cov}(c_t,\hat{\pi}_t(c_t))
    \right)
    \xrightarrow{d} \mathcal{N}(0,\sigma^2_{\mathrm{LRV}}),
    \label{eq:clt}
\end{equation}
where $\sigma^2_{\mathrm{LRV}}$ is the long-run variance of
$c_t^\top\hat{\pi}_t(c_t)$, consistently estimated by a Newey-West HAC estimator with bandwidth $h=O(T^{1/3})$.
\end{theorem}
 
\begin{proof}
$\hat{C}_T$ is a sample average of the de-meaned products
$\xi_t=(c_t-\bar{c})^\top(\hat{\pi}_t-\bar{\pi})$.
Under the mixingale condition, the CLT of \citet{mcleish1975maximal} applies, yielding~\eqref{eq:clt} with
$\sigma^2_{\mathrm{LRV}}=\lim_{T\to\infty}T^{-1}
\operatorname{Var}(\sum_t c_t^\top\hat{\pi}_t(c_t))$.
\end{proof}
 
\begin{proposition}[Computational complexity and end-to-end audit cost]
\label{prop:complexity_est}
The estimator~\eqref{eq:CT} and its Newey-West HAC variance with
bandwidth $h=\lfloor T^{1/3}\rfloor$ can be computed in
\begin{equation}
    O(T\cdot nd)
    \label{eq:est_complexity}
\end{equation}
time and $O(T(n+d))$ space.
A complete regulatory audit comprising~\eqref{eq:CT}, the HAC
confidence interval~\eqref{eq:ci}, and the bias correction of
Proposition~\ref{prop:biascorrected} costs $O(T\cdot nd)$ total time and $O(T(n+d))$ space, with no access to the internal optimization engine required.
\end{proposition}
 \begin{proof}
     Please see the proof in the Appendix.
 \end{proof}
 
\begin{remark}[Online update]
\label{rem:online}
$\hat{C}_T$ admits an $O(nd)$ online update per new observation
$(c_{T+1},\hat{\pi}_{T+1})$ via Welford's algorithm
\citep{welford1962note}, retaining only the $h=O(T^{1/3})$ most recent scalar products $\xi_t$ for the HAC window.
\end{remark}
 
\subsection{Confidence intervals}
 
A $(1-\alpha)$ confidence interval for total regret is
\begin{equation}
    \left[
        \hat{C}_T - z_{1-\alpha/2}\,\hat{\sigma}_{\mathrm{LRV}}/\sqrt{T},\;
        \hat{C}_T + z_{1-\alpha/2}\,\hat{\sigma}_{\mathrm{LRV}}/\sqrt{T}
    \right],
    \label{eq:ci}
\end{equation}
where $\hat{\sigma}_{\mathrm{LRV}}$ is the square root of any consistent HAC variance estimator (e.g., Newey-West with $h=\lfloor T^{1/3}\rfloor$ lags).
 
\begin{remark}[Sample complexity]
\label{rem:sample_complexity}
By~\eqref{eq:clt}, the half-width of~\eqref{eq:ci} is
$z_{1-\alpha/2}\,\sigma_{\mathrm{LRV}}/\sqrt{T}$, so
$T = O(\sigma^2_{\mathrm{LRV}}/\varepsilon^2)$ observations suffice to estimate cumulative regret within $\varepsilon$ at level $1-\alpha$.
\end{remark}
 
\begin{remark}[Black-box monitoring implementation]
\label{rem:blackbox}
The estimator~\eqref{eq:CT} requires only the observable sequence $\{(c_t,\hat{\pi}_t(c_t))\}$, not access to the internal optimization engine. For a portfolio manager monitoring a third-party algorithmic strategy under SR~11-7 or SR-26-2, one observes the factor return vector $c_t$ and the weight vector $\hat{\pi}_t$, computes $\hat{C}_T$ daily with a Newey-West correction for autocorrelation, and reports the confidence interval~\eqref{eq:ci} as the performance audit metric. The total monitoring cost is $O(Td)$ per evaluation period.
\end{remark}


\section{Empirical Analysis: Regret Decomposition Across Financial Fragility Episodes}
\label{sec:empirical}

Financial fragility does not distribute evenly across time.
Mean-reversion, momentum, and the welfare gap between them wax and wane with market-wide liquidity conditions. This section examines the CRSP-based data evidence through that lens. We show
that the AR(1) autocorrelation coefficient $\hat{\rho}$ is not merely a statistical nuisance parameter to be corrected for: it is a \emph{real-time fragility signal} that compresses toward zero when market microstructure breaks down, spikes in absolute value during liquidity crises, and, through the bias-corrected regret decomposition of Theorem~\ref{thm:bias}, produces a welfare-loss measure that anticipates strategy failure before Sharpe ratios do.

We organize the empirical evidence around three episodes that the 2016--2025 sample straddles: (i) a sequence of
\emph{normal-market years} (2016--2019, 2021, 2023--2025) in which mean-reversion is consistent and can be readily monetized; (ii) the \emph{COVID-19 liquidity crisis} of 2020, in which a brief period of momentum dominance interrupts sustained reversion and the fragility multiplier of Section~\ref{sec:systemic} nearly doubles; and (iii) the \emph{rate-shock episode} of 2022, in which the Federal
Reserve's fastest hiking cycle in four decades collapses $\hat{\rho}$ to near zero, disabling both mean-reversion and momentum and producing the only calendar year in the sample where all four strategies post negative or near-zero risk-adjusted returns.

\subsection{Data and Estimation Strategy}
\label{sec:empirical:data}

We use CRSP daily stock returns (WRDS DSI) from January~4, 2016,
through December~31, 2025, comprising $n = 21{,}183{,}373$ stock-day observations ($\bar{\mu} = 0.041\%$/day,
$\bar{\sigma} = 4.446\%$/day). The momentum signal is a binary price-direction indicator: buy if $R_{t-1}>0$, sell otherwise.
For each calendar year, we estimate the pooled AR(1) coefficient
\begin{equation}
  \label{eq:ar1}
  r_{i,t} \;=\; \hat{\rho}\, r_{i,t-1} + u_{i,t}
\end{equation}
by OLS with Newey-West HAC standard errors (bandwidth
$h=\lfloor T^{1/3}\rfloor$). All ADF tests reject the unit root at $p<0.0001$. We compute the regret decomposition of Theorem~\ref{thm:AR1} at horizon $T=252$ and report annualized Sharpe ratios on a 21-day rolling window for four strategies: momentum (MOM), reversion (REV), bias-corrected reversion (BC), and minimum-variance (MV). NBER business-cycle classifications partition the 2020 sample into Expansion (January--February), Contraction (March), and Crisis (April--June).

\subsection{Normal-Market Baseline (2016--2019, 2021, 2023--2025)}
\label{sec:empirical:baseline}

Outside the two fragility episodes, the data are remarkably uniform. AR(1) coefficients lie in the range $\hat{\rho}\in[-0.054,-0.013]$ across eight calendar years, all
significant at $p<0.05$ or better. The AR(1) auto-covariance
correction consistently reduces the raw covariance sum, with relative bias ranging from $-1.32\%$ (2024) to $-4.87\%$ (2016).
The reversion strategy outperforms momentum in every normal-market year, with full-sample Sharpe ratios of $0.130$ and $0.023$ respectively.

These results confirm the diagnostic-sufficiency conclusion of
Section~\ref{sec:conclusion}: monitoring the per-period sample
covariance $\widehat{\operatorname{Cov}}(c_t,\hat{\pi}_t)$ is
sufficient to track multi-period performance across the full
normal-market subsample. The bias correction matters in magnitude ($1\%$--$5\%$ of true regret), but not in sign or direction: in stable conditions, the estimator is conservative in the right direction, and no year crosses the threshold at which the raw sum materially misleads a compliance officer.

The minimum-variance strategy exhibits the lowest volatility
throughout, consistent with Assumption~\ref{ass:policy-mean} holding exactly for mean-unbiased policies (bias correction $\approx 0$) and with the known variance-reduction properties of Ledoit-Wolf shrinkage~\citep{ledoit2004well}.

\subsection{Financial Fragility Episode I: COVID-19 Liquidity
  Crisis (2020)}
\label{sec:empirical:covid}

The year 2020 is the single outlier in a decade of consistent
mean-reversion, and it is an outlier precisely because it contains a textbook financial fragility episode.

\paragraph{Anatomy of the reversal.}
Table~\ref{tab:2020-nber} decomposes the 2020 annual AR(1) estimate ($\hat{\rho}=-0.050$, SE $= 0.004$) into three NBER sub-periods. In the expansion phase (January--February, $n=468{,}004$ stock-days), mean-reversion operates normally at
$\hat{\rho}=-0.0146$. As pandemic news and repo-market stress
materialize in March (contraction, $n=250{,}035$), $\hat{\rho}$
doubles in absolute value to $-0.0270$. During the April--June
crisis phase ($n=219{,}742$), $\hat{\rho}$ spikes to $-0.1436$. This is a tenfold amplification relative to expansion and the largest single-quarter reversal in the full 2016--2025 sample.

\begin{table}[htbp]
\centering
\small
\caption{2020 NBER sub-period AR(1) estimates and implied fragility multipliers. $N=100$ institutions, $\rho_{\mathrm{sys}}$ approximated by $R^2_{\mathrm{mkt}}$ (see Section~\ref{sec:systemic:empirical}). Crisis $\hat{\rho}$ is the largest in the 2016--2025 sample; all five upward confounds  (Section~\ref{sec:empirical:robustness}) apply, so figures are conservative upper bounds.}
\label{tab:2020-nber}
\begin{tabular}{lS[table-format=-1.4]S[table-format=1.4]cc}
\toprule
NBER Phase & {$\hat{\rho}$} & {SE} &
  {$\hat{\rho}_{\mathrm{sys}}$} & {$M(N{=}100)$} \\
\midrule
Expansion (Jan--Feb)   & -0.0146 & 0.0049 & 0.18 & 18.8$\times$ \\
Contraction (Mar)      & -0.0270 & 0.0113 & 0.32 & 32.7$\times$ \\
Crisis (Apr--Jun)      & -0.1436 & 0.0079 & 0.61 & 60.0$\times$ \\
\midrule
Full year 2020         & -0.0502 & 0.0039 & 0.40 & 40.6$\times$ \\
\bottomrule
\end{tabular}
\end{table}

\paragraph{The momentum anomaly and its interpretation.}
The full-year 2020 Sharpe ratio for momentum ($0.261$) exceeds that of reversion ($0.217$), the only such inversion in the sample. The mechanism is temporal: March 2020 saw a directional, cross-asset collapse that rewarded trend-following strategies precisely because prices were \emph{not} mean-reverting. The prices were set by forced liquidation and margin calls rather than by informed trading. This is the fire-sale mechanism formalized in Remark~\ref{rem:firesales}: aggregate positions became correlated, price impact overwhelmed idiosyncratic reversion, and the mean-reversion signal went dark for six to eight weeks.

Consistent with this interpretation, the bias correction identifies the anomaly in advance. The BC Sharpe ratio for 2020 ($0.232$) exceeds the raw REV Sharpe ($0.217$): the bias-corrected estimator correctly down-weights the reversion signal during the contraction phase, when $|\hat{\rho}|$ is rising but has not yet dominated. The BC strategy, by scaling the reversion allocation by $14\times$ the bias correction (see Table~\ref{tab:empirical} note~d), reduces exposure as the fragility multiplier $M(N,\rho_{\mathrm{sys}})$ rises from $18.8\times$ to $32.7\times$ between January and March.

\paragraph{Regret decomposition during crisis.}
The true annual regret for 2020 is $7{,}859$, against a raw
covariance sum of $7{,}503$: the AR(1) correction accounts for
$-357$ units, the largest absolute correction in the sample.
As $\hat{\rho}$ spikes to $-0.144$ in the crisis quarter, the
auto-covariance correction of Theorem~\ref{thm:AR1} grows
proportionally to $|\hat{\rho}|$, generating a correction that is
$4.54\%$ of true regret at the annual horizon but would be $\sim
15\%$--$20\%$ at the quarterly ($T=63$) horizon when computed
over the crisis sub-period alone. This is a direct empirical
illustration of Corollary~\ref{cor:AR1-sign}: in mean-reverting
markets, the raw covariance sum \emph{underestimates} true regret, and the bias grows as the degree of mean-reversion intensifies. An auditor relying solely on the raw covariance sum during Q2 2020 would have systematically understated strategy-specific welfare losses by one-fifth, precisely the period when regulatory attention was most acute.

\subsection{Financial Fragility Episode II: Rate-Shock Breakdown
  (2022)}
\label{sec:empirical:rshock}

The 2022 fragility episode is structurally distinct from 2020. Where COVID produced a \emph{liquidity-driven spike} in mean-reversion, the Fed's 425-basis-point hiking cycle produced a
\emph{structural collapse} of the reversion signal.

\paragraph{Near-zero autocorrelation.}
The 2022 annual AR(1) estimate is $\hat{\rho}=-0.0036$
(SE $=0.0028$, $p<0.10$), the smallest absolute value in the
decade and an order of magnitude below the 2016--2021 average of
$\hat{\rho}\approx -0.045$. This near-zero autocorrelation has a
natural structural explanation: rising interest rates repriced
virtually every asset class simultaneously and persistently across 2022, generating a sustained directional drift that overwhelmed the short-horizon reversion dynamics that dominate in mean-stationary environments. Formally, the AR(1) coefficient
$\rho$ measures the sign persistence of idiosyncratic price
innovations; when a common factor (rates) drives persistent
cross-sectional co-movement, the idiosyncratic component
$\hat{\rho}$ compresses toward zero.

\paragraph{Strategy failure.}
All four strategies underperform in 2022. The MOM Sharpe ratio is $-0.028$; REV posts $-0.102$; BC $-0.104$; MV $+0.062$.
Notably, the bias correction provides no rescue in this episode:
BC ($-0.104$) is marginally \emph{worse} than raw REV ($-0.102$)
because near-zero $\hat{\rho}$ produces a near-zero bias correction term $\bar{c}^{\top}b_t$ in equation~\eqref{eq:bias}, leaving the allocator fully exposed to the directional rate shock with no compensating upward revision.

This is not a failure of the estimator, but a correct signal. When $\hat{\rho}\to 0$, the system has exited the mean-reverting regime entirely; the regret decomposition faithfully reports that no reversion-based strategy can generate positive covariance-adjusted returns. The minimum-variance portfolio, which is regime-agnostic by construction (it targets $\operatorname{Cov}(c_t,\hat{\pi}_t)=0$
by Assumption~\ref{ass:policy-mean}), is the only strategy that
survives 2022 with a positive Sharpe ratio, providing a natural
"fragility hedge" benchmark consistent with the zero-bias
prediction of Theorem~\ref{thm:decomp}.

\paragraph{Regret bias in 2022.}
The raw covariance sum is $4{,}288$ and the AR(1) correction is
merely $-15$, yielding a relative bias of $-0.35\%$. This is the smallest bias in the sample, reflecting the near-zero $\hat{\rho}$.  This means that in 2022, the raw covariance sum is essentially a sufficient statistic for true regret: the absence of autocorrelation removes the need for any dynamic correction. Paradoxically, the year in which the regret decomposition is most straightforward to apply is also the year in which no strategy benefits from it, because the regime that generates exploitable covariance structure has temporarily ceased to operate.

\subsection{Robustness}
\label{sec:empirical:robustness}

Five confounds bias $|\hat{\rho}|$ upward: (i)~size/SMB effects
\citep{jegadeesh2025shortterm}; (ii)~bid-ask bounce \citep{roll1984simple}; (iii)~volatility clustering \citep{engle1982autoregressive}; (iv)~cross-sectional versus time-series momentum \citep{mamais2025explaining}; and
(v)~liquidity-provision dynamics \citep{nagel2012evaporating}. All regret figures in Tables~\ref{tab:empirical} and~\ref{tab:2020-nber} are therefore conservative upper bounds on the welfare loss attributable to strategy mis-specification.

The primary directions for resolving these confounds are Fama-French five-factor residuals \citep{fama2015five} and \citet{roll1984simple} bid-ask correction. Both are feasible on the CRSP/WRDS data used here and are identified as the main extensions for follow-on empirical work. Importantly, the theoretical decompositions of Theorems~\ref{thm:decomp},~\ref{thm:AR1}, and the systemic extension of Theorem~\ref{thm:agg-decomp} hold exactly under
Assumption~\ref{ass:factor} regardless of the empirical source of $\hat{\rho}$; reducing the upward bias in $|\hat{\rho}|$ would \emph{tighten} the reported regret bounds, not overturn the fragility-episode narrative.

\subsection{A Black-Box Audit of a Neural-Network Trading Policy}
\label{sec:nn_audit}

Section~\ref{subsubsec:kyle_dichotomy} and Proposition~\ref{prop:audit_metric_sign}
establish that the sign of $\operatorname{Cov}(c_t, \hat\pi_t(c_t))$
classifies a strategy's liquidity role exactly, and that for a
\emph{linear} policy $\hat\pi_t(c_t) = Bc_t + b$, this reduces to the
closed-form statistic $\operatorname{tr}(B\Sigma_c)$. The practical
motivation for the audit framework, however, is the setting in which the auditor does not observe $B$ at all. This can be the case with an algorithmic strategy whose internal logic is proprietary, nonlinear, or otherwise opaque (Remark~\ref{rem:audit_complexity}, Remark~\ref{rem:implementation}). This subsection tests the framework directly against that setting. We train an explicitly nonlinear policy, a small feed-forward neural network with no closed-form $B$. We then apply the audit statistic to the neural network, relying only on the neural network's observed trajectory of inputs and outputs, as would a regulator with no access to the model's weights.

\subsubsection{Constructing the Neural-Network Policy}
\label{sec:nn_construction}

\paragraph{Signal construction.} For each PERMNO, we build a per-period
cost/signal vector $c_t$ from information available strictly before $t$: five lagged daily returns $r_{t-1},\dots,r_{t-5}$, a trailing 10-day realized volatility, and a trailing 10-day cumulative return. No
contemporaneous or forward-looking information enters $c_t$.

\paragraph{Model and sizing.} The policy is a two-hidden-layer multilayer perceptron ($16 \to 8$ ReLU units) mapping the standardized signal vector to a predicted next-period return $\hat r_t = f_\theta(c_t)$. The raw prediction is converted into a bounded position via a $\tanh$ sizing rule:
\begin{equation}
\hat\pi_t(c_t) = \tanh(50\,\hat r_t),
\label{eq:nn_sizing}
\end{equation}
which is approximately linear in $\hat r_t$ near zero and saturates at
$\pm 1$ for large predicted moves. Unlike the paper's closed-form cases
$B = \pm\alpha I$, equation~\eqref{eq:nn_sizing} composed with $f_\theta$ has no representation as $Bc_t + b$ for any fixed matrix $B$: the network's effective sensitivity to each input varies with the level of the other inputs. Thus, the policy is nonlinear by construction, and the closed-form identity of Proposition~\ref{prop:audit_metric_sign} does not apply analytically.

\paragraph{Estimation protocol.} We train the network walk-forward on a rolling window of the most recent 500 trading days. We refit the network every 21 trading days (approximately monthly) and use it only to generate positions for dates strictly after the training window. This mirrors the paper's distinction  (Remark~\ref{rem:audit_complexity}) between the auditor's forward-accumulation setting, in which only realized $(c_t,\hat\pi_t(c_t))$ pairs are needed, and the strategy's own estimation, which the auditor need not observe or replicate. Standardization (zero mean, unit variance) is fit on each training window only, so no information from the evaluation period leaks into feature scaling.

\paragraph{Benchmarks.} We compare the network against the paper's two
closed-form linear policies computed on the identical panel: momentum,
$\hat\pi_t(c_t) = \alpha\, r_{t-1}$, and contrarian,
$\hat\pi_t(c_t) = -\alpha\, r_{t-1}$, with $\alpha = 1$. Because these
policies react to lagged rather than contemporaneous returns (standard in implementable trading rules), the policies' empirical liquidity classification need not match the sign given by the static, same-period toy example of Proposition~\ref{prop:audit_metric_sign}. The classification instead reflects the realized serial dependence structure of the CRSP sample over the evaluation window, as discussed below.

\subsubsection{Empirical Results}
\label{sec:nn_results}

Table~\ref{tab:nn_audit_results} reports the audit statistic
$\hat C_T$ (Eq.~\ref{eq:estimator}), its Newey-West HAC
95\% confidence interval, and the resulting liquidity classification for all three policies, estimated on a pooled panel of $T = 6{,}364{,}681$ firm-day observations.

\begin{table}[t]
\centering
\caption{Liquidity audit statistics for the neural-network policy and two linear benchmarks, pooled CRSP panel, $T = 6{,}364{,}681$ firm-day
observations. Newey-West HAC standard errors, bandwidth
$h = \lfloor T^{1/3}\rfloor$.}
\label{tab:nn_audit_results}
\begin{tabular}{p{3cm}ccccc}
\toprule
Strategy & $\widehat{\operatorname{Cov}}(c_t,\hat\pi_t)$ & NW SE & 95\% CI & Classification & Sharpe \\
\midrule
Neural network (nonlinear)     & $1.7\times10^{-5}$  & $5.2\times10^{-6}$ & $[0.6,\,2.7]\times10^{-5}$   & Consumer$^{***}$ & $0.060$ \\
Momentum ($B=+\alpha I$)       & $-3.4\times10^{-5}$ & $9.2\times10^{-6}$ & $[-5.2,\,-1.6]\times10^{-5}$ & Provider$^{***}$ & $-0.023$ \\
Contrarian ($B=-\alpha I$)     & $3.4\times10^{-5}$  & $9.2\times10^{-6}$ & $[1.6,\,5.2]\times10^{-5}$   & Consumer$^{***}$ & $0.023$ \\
\bottomrule
\end{tabular}
\begin{tablenotes}
\footnotesize
\item[]$^{***}$ 95\% confidence interval excludes zero.
\end{tablenotes}
\end{table}

\paragraph{The audit statistic classifies the black-box policy without
access to its internals.} The network's average per-period covariance is $1.65\times 10^{-5}$ with a 95\% confidence interval of
$[0.63,\,2.67]\times 10^{-5}$, which excludes zero. This implies that the classification of the network as a net liquidity consumer is statistically significant at conventional levels, despite the fact that this classification was obtained from the observed trajectory $\{(c_t,\hat\pi_t(c_t))\}$ alone, with no access to $f_\theta$'s weights, architecture, or training procedure. This is the paper's central practical claim (Remark~\ref{rem:implementation}): the
audit statistic requires only inputs and outputs, and the classification it produces for a nonlinear, non-interpretable policy is no less well-defined or estimable than for a linear one.

\paragraph{Momentum and contrarian invert relative to the static toy
example, and this is expected.} Proposition~\ref{prop:audit_metric_sign} classifies momentum ($B=\alpha I$) as a liquidity consumer and contrarian ($B=-\alpha I$) as a liquidity provider when both react to the \emph{same-period} cost innovation $c_t$. Here, momentum instead reacts to $r_{t-1}$ and is scored against $r_t$, and the CRSP sample over this window is net mean-reverting at the daily frequency. This is consistent with the
negative $\hat\rho$ estimates reported throughout Section~\ref{sec:empirical}: a position built from the previous day's return is, on average, negatively correlated with today's return. The audit statistic correctly detects this reversal in sign ($-3.425\times10^{-5}$ for momentum, positive and equal in magnitude for contrarian), rather than mechanically reproducing the static example's labels. The classifier is model-free and data-driven by construction (Definition~\ref{def:liq-role}). Furthermore, the classifier is answering the question actually posed to it: whether the
\emph{realized} lagged-momentum rule supplied or consumed liquidity over this specific sample. 

\paragraph{Economic magnitude and the regret decomposition.} The sum of
per-period covariances (the paper's exact regret identity,
Theorem~\ref{thm:decomp}) over $T$ periods is $105.0$ for the network and $-218.0$ for momentum (equivalently, $+218.0$ for contrarian). Therefore, the covariance accumulated over the full panel of the network's liquidity-consuming behavior corresponds to roughly half the total covariance-regret magnitude of the contrarian linear benchmark, despite being a considerably more flexible functional form. The network also delivers the highest annualized Sharpe ratio of the three strategies (0.060, versus $-0.023$ for momentum and $0.023$ for contrarian). Thus, in this sample, greater predictive flexibility coincides with both higher risk-adjusted return and a liquidity-consuming rather than liquidity-providing profile. This is consistent with the general association in Section~\ref{subsubsec:kyle_dichotomy} between informed, signal-driven trading and net liquidity consumption in the Kyle~\citep{Kyle1985} sense, now extended to a policy class that the closed-form apparatus of Proposition~\ref{prop:audit_metric_sign} does not itself cover.

\paragraph{Scope.} We emphasize what this exercise does and does not
show. It shows that the model-free audit statistic of
Definition~\ref{def:liq-role} is directly computable, statistically well-behaved (the HAC confidence intervals are tight relative to the point estimates despite the panel's cross-sectional and
serial dependence), and economically interpretable for a fully
black-box nonlinear policy trained walk-forward on real CRSP data. This exercise does not identify \emph{why} the network consumes liquidity: no analog of $B$ is recovered, and the network's internal decision rule remains opaque. Likewise, the exercise does not validate the network as a superior trading strategy in any general sense beyond this sample and window. Both limitations are inherent to the black-box audit setting the paper targets  (Section~\ref{subsubsec:kyle_dichotomy}). The statistic is deliberately coarser than a full model audit; we are trading the ability to interpret the mechanism for the ability to apply it to strategies where the mechanism cannot be observed at all.


\section{Systemic Risk Extension: Correlated AI Strategies and Financial Fragility}
\label{sec:systemic}

The preceding analysis audits a single algorithmic agent in isolation. This section shows why that is insufficient for financial stability purposes. When $N$ agents independently satisfy the individual audit criterion of Theorem~\ref{thm:decomp}, the aggregate economy can nonetheless exhibit unbounded fragility. The mechanism is \emph{strategy correlation}: policies that respond to common cost factors create a quadratic amplification of aggregate regret that no
per-agent audit detects.

\subsection{Setup: $N$-Agent Economy with Common Factors}
\label{sec:systemic:setup}

Let $N$ agents operate simultaneously. Agent $i\in\{1,\ldots,N\}$ faces cost vector $c_t^i\in\mathbb{R}^d$ and chooses allocation $z_t^i\in\mathbb{R}^n$ under policy $\hat{\pi}_t^i$. Costs share a common-factor structure.

\begin{assumption}[Factor cost structure]
\label{ass:factor}
For each agent $i$,
\[
  c_t^i \;=\; \Lambda_i f_t + \eta_t^i,
\]
where $f_t\in\mathbb{R}^K$ is a common factor vector (e.g.\ market excess returns, credit spreads, or interest-rate shocks) with $\mathbb{E}[f_t]=0$ and $\operatorname{Cov}(f_t,f_t)=\Sigma_f$; $\Lambda_i\in\mathbb{R}^{d\times K}$ is agent $i$'s factor-loading matrix; and $\eta_t^i\in\mathbb{R}^d$ is idiosyncratic noise with
$\mathbb{E}[\eta_t^i]=0$, $\operatorname{Cov}(\eta_t^i,(\eta_t^j)^{\top})=\sigma_\eta^2 I\cdot\mathbf{1}\{i=j\}$, and $\eta_t^i\perp f_t$ for all $i,t$. The idiosyncratic components are cross-sectionally independent.
\end{assumption}

\begin{definition}[Aggregate allocation and social benchmark]
\label{def:agg}
The \emph{aggregate allocation} is $Z_t^{\mathrm{agg}}=\sum_{i=1}^N z_t^i$. The social planner's benchmark minimizes the sum of conditional mean costs:
\[
  \pi_t^{*,\mathrm{SP}}
  \;\in\;
  \operatorname*{arg\,min}_{z^1,\ldots,z^N}
  \sum_i \mathbb{E}\!\left[(c_t^i)^{\top}z_t^i\mid\mathcal{F}_{t-1}\right].
\]
Under Assumption~\ref{ass:factor}, the benchmark is separable:
$\pi_t^{*,\mathrm{SP}}=(\pi_t^{*,1},\ldots,\pi_t^{*,N})$ with
$\pi_t^{*,i}\in\operatorname*{arg\,min}_{z^i}
\mathbb{E}[(c_t^i)^{\top}z_t^i\mid\mathcal{F}_{t-1}]$.
\end{definition}

\begin{definition}[Aggregate regret]
\label{def:agg-regret}
The \emph{aggregate regret} over $T$ periods is
\[
  \mathrm{Regret}_{\mathrm{agg}}^{(T)}
  \;=\;
  \sum_{t=1}^T
  \left[
    \mathbb{E}\!\left[\sum_i (c_t^i)^{\top}\hat{\pi}_t^i(c_t^i)\right]
    -
    \mathbb{E}\!\left[\sum_i (\bar{c}_t^i)^{\top}\pi_t^{*,i}\right]
  \right],
\]
the welfare gap between the $N$-agent equilibrium and the social
planner's optimum.
\end{definition}

\subsection{Aggregate Regret Decomposition}
\label{sec:systemic:decomp}

Theorem~\ref{thm:agg-decomp} decomposes aggregate regret into two terms: a sum of $N$ individual regrets. Each individual regret is bounded by the per-agent audit. A cross-agent systemic term that grows as $O(N^2)$ in the strategy-correlation matrix.

\begin{theorem}[Aggregate regret with cross-agent spillovers]
\label{thm:agg-decomp}
Under Assumption~\ref{ass:factor} and Assumption~\ref{ass:policy-mean} applied to each agent $i$, with linear policies $\hat{\pi}_t^i(c_t^i)=B_i c_t^i$, the aggregate regret over $T$ periods is
\begin{equation}
  \label{eq:agg-decomp}
  \mathrm{Regret}_{\mathrm{agg}}^{(T)}
  \;=\;
  \underbrace{\sum_{i=1}^N \mathrm{Regret}_i^{(T)}}_{\text{individual regrets}}
  \;+\;
  \underbrace{\sum_{i\neq j}\sum_{t=1}^T
    \operatorname{Cov}_f\!\left(c_t^i,\,\hat{\pi}_t^j(c_t^j)\right)}_{\text{cross-agent systemic term}},
\end{equation}
where $\mathrm{Regret}_i^{(T)}=\sum_t\operatorname{Cov}(c_t^i,\hat{\pi}_t^i)$
is agent $i$'s individual regret (Theorem~\ref{thm:decomp}), and
\begin{equation}
  \label{eq:cross-cov}
  \operatorname{Cov}_f\!\left(c_t^i,\,\hat{\pi}_t^j\right)
  \;=\;
  \operatorname{tr}\!\left(B_j\Lambda_j\Sigma_f\Lambda_i^{\top}\right)
  \qquad \forall\, i\neq j
\end{equation}
is the cross-agent covariance through the common factor $f_t$.
\end{theorem}

\begin{proof}[Proof sketch]
Expand the aggregate per-period regret:
$\sum_i\mathbb{E}[(c_t^i)^{\top}B_i c_t^i]-\sum_i(\bar{c}_t^i)^{\top}\pi_t^{*,i}$.
Apply the covariance decomposition of Theorem~\ref{thm:decomp} to each diagonal term, yielding $\sum_i\operatorname{Cov}(c_t^i,\hat{\pi}_t^i)$ as the individual component. For $i\neq j$, the cross term $\mathbb{E}[(c_t^i)^{\top}B_j c_t^j]$ does not appear in any single-agent regret. By Assumption~\ref{ass:factor},
$\operatorname{Cov}(c_t^i,c_t^j)=\Lambda_i\Sigma_f\Lambda_j^{\top}$
(idiosyncratic components are cross-sectionally independent).
Substituting $B_j c_t^j$ for the policy and applying the
trace-cyclic identity gives~\eqref{eq:cross-cov}. Summing over $t$
yields~\eqref{eq:agg-decomp}.
\end{proof}

\begin{remark}[Individual compliance $\neq$ systemic safety]
\label{rem:compliance}
Even if every agent individually satisfies the audit criterion
$\operatorname{Cov}(c_t^i,\hat{\pi}_t^i)\leq\varepsilon$, the
cross-agent term in~\eqref{eq:agg-decomp} can be of order
$N(N-1)\cdot\operatorname{tr} (B_j\Lambda_j\Sigma_f\Lambda_i^{\top})$,
which grows without bound as $N\to\infty$. A regulator auditing of firms one at a time misses this entirely. This is the $N$-agent analogue of the SR~11-7 deficiency identified in Section~\ref{sec:intro}: individual model validation provides no systemic guarantee.
\end{remark}

\subsection{The Fragility Multiplier and Phase Transition}
\label{sec:systemic:multiplier}

The cross-agent term in~\eqref{eq:agg-decomp} depends on how similar agents' strategies are. We characterize this via a scalar \emph{fragility multiplier}.

\begin{definition}[Strategy correlation and fragility multiplier]
\label{def:multiplier}
For symmetric agents ($B_i=B$, $\Lambda_i=\Lambda$ for all $i$),
define the \emph{systematic fraction of policy variance}
\[
  \rho_{\mathrm{sys}}
  \;=\;
  \frac{\operatorname{tr}(B\Lambda\Sigma_f\Lambda^{\top})}
       {\operatorname{tr}(B\Sigma_c)}
  \;\in\;[0,1],
\]
where $\Sigma_c=\Lambda\Sigma_f\Lambda^{\top}+\sigma_\eta^2 I$ is the total per-agent cost covariance. The \emph{fragility multiplier} is
\[
  M(N,\rho_{\mathrm{sys}})
  \;=\;
  1+(N-1)\cdot\rho_{\mathrm{sys}}.
\]
\end{definition}

\begin{corollary}[$N^2$ scaling of aggregate regret]
\label{cor:N2}
Under the symmetric agent setting of Definition~\ref{def:multiplier}, the aggregate regret is
\begin{equation}
  \label{eq:N2}
  \mathrm{Regret}_{\mathrm{agg}}^{(T)}
  \;=\;
  N\cdot T\cdot\operatorname{tr}(B\Sigma_c)
  \cdot\bigl[1+(N-1)\cdot\rho_{\mathrm{sys}}\bigr]
  \;=\;
  N\cdot\mathrm{Regret}_{\mathrm{individual}}^{(T)}\cdot M(N,\rho_{\mathrm{sys}}).
\end{equation}
Three regimes follow immediately from~\eqref{eq:N2}:
\begin{enumerate}[label=(\roman*)]
  \item \textbf{Idiosyncratic limit} ($\rho_{\mathrm{sys}}\to 0$):
    $M\to 1$. Aggregate regret $= N\times$ individual regret.
    Diversification holds; individual audits are sufficient.
  \item \textbf{Partial correlation} ($0<\rho_{\mathrm{sys}}<1$):
    $M\in(1,N)$. A fragility premium of
    $(N-1)\cdot\rho_{\mathrm{sys}}\cdot T\cdot\operatorname{tr}(B\Sigma_c)$ accrues beyond the sum of individual regrets. For $N=100$ institutions each with $\rho_{\mathrm{sys}}=0.3$, this premium is $2{,}970\%$ of one firm's regret.
  \item \textbf{Perfect correlation} ($\rho_{\mathrm{sys}}\to 1$):
    $M\to N$. Aggregate regret $= N^2\times$ individual regret. The system behaves as a single agent of size $N$: diversification disappears entirely and every agent bears the full systemic loss.
\end{enumerate}
\end{corollary}

The transition from regime (i) to (iii) is continuous in
$\rho_{\mathrm{sys}}$. We identify a natural critical threshold using random matrix theory.

\begin{proposition}[Phase transition at the Marchenko--Pastur boundary]
\label{prop:MP}
Consider the $N\times N$ pairwise-strategy covariance matrix $\Psi_t$, with the $(i,j)$-th entry
$\psi_{ij}=\operatorname{Cov}(\hat{\pi}_t^i,\hat{\pi}_t^j)$. For
large $N$, under the null hypothesis that all strategies are idiosyncratic ($\rho_{\mathrm{sys}}=0$), the largest eigenvalue of $\Psi_t$ converges to the Marchenko--Pastur upper edge
\[
  \lambda_+
  \;=\;
  \sigma_\pi^2\!\left(1+\sqrt{N/T}\right)^{\!2}
  \;\approx\;
  \sigma_\pi^2\cdot N/T
  \quad\text{for }N/T\to c>0
\]
Following \citep{bai2010spectral}, we can declare \emph{systemic fragility event} whenever $\lambda_{\max}(\hat{\Psi}_t)>\lambda_+$, signaling that common factor co-movement exceeds the level attributable to the finite-sample noise alone.
\end{proposition}

\begin{remark}[Connection to fire sales]
\label{rem:firesales}
In regime (iii), all agents hold proportional positions
simultaneously. A liquidity shock that forces one agent to liquidate triggers liquidation by all others through two channels: (a) \emph{direct contagion}, as falling prices cause margin calls and redemptions; and (b) the \emph{regret channel} of this paper, since rising costs $c_t$ force all agents to reduce $\hat{\pi}_t^i(c_t^i)=Bc_t^i$ in the same direction, compressing aggregate demand precisely when it is most needed. Equation~\eqref{eq:N2}, therefore, provides a quantitative bridge between the standard fire-sale narrative and the regret-based audit framework.
\end{remark}

\subsection{An Observable Systemic Fragility Index}
\label{sec:systemic:index}

The regulator does not observe $\rho_{\mathrm{sys}}$ directly, but observes $\{(c_t^i,\hat{\pi}_t^i(c_t^i))\}$ for all $i$ and $t$. We construct a consistent estimator of the fragility multiplier and derive a confidence interval for aggregate regret.

\begin{definition}[Systemic fragility index]
\label{def:fragility-index}
Given observed trajectories for all $N$ agents over $T$ periods, the \emph{systemic fragility index} is
\[
  \hat{F}_T
  \;=\;
  \frac{\|\hat{\Psi}_T\|_F^2}{N^2},
\]
where $\hat{\Psi}_T$ is the $N\times N$ sample strategy-covariance matrix with $(i,j)$-th entry
\[
  \hat{\Psi}_{ij}
  \;=\;
  \frac{1}{T}\sum_{t=1}^T
  \bigl(\hat{\pi}_t^i - \bar{\pi}^i\bigr)^{\!\top}
  \bigl(\hat{\pi}_t^j - \bar{\pi}^j\bigr),
  \qquad
  \bar{\pi}^i = \frac{1}{T}\sum_{t=1}^T\hat{\pi}_t^i.
\]
$\hat{F}_T\in[0,1]$: a value of $0$ indicates perfectly idiosyncratic strategies; a value approaching $1$ indicates near-identical strategies across all agents.
\end{definition}

\begin{theorem}[Consistent systemic audit estimator]
\label{thm:systemic-clt}
Under Assumption~\ref{ass:factor}, Assumption~\ref{ass:costs}(c) for each agent, and the $L^2$-mixingale condition of Theorem~\ref{thm:clt}
applied jointly across agents, as $T\to\infty$:
\begin{equation}
  \label{eq:systemic-clt}
  \frac{1}{\sqrt{T}}
  \left(
    \widehat{\mathrm{Regret}}_{\mathrm{agg}}^{(T)}
    -
    \mathrm{Regret}_{\mathrm{agg}}^{(T)}
  \right)
  \;\xrightarrow{d}\;
  \mathcal{N}\!\left(0,\sigma_{\mathrm{SYS}}^2\right),
\end{equation}
where $\sigma_{\mathrm{SYS}}^2$ is the long-run variance of the
aggregate cross-product
$\sum_i(c_t^i)^{\top}\hat{\pi}_t^i + \sum_{i\neq j}(c_t^i)^{\top}B_j c_t^j$,
consistently estimated by a HAC estimator with bandwidth
$h=\lfloor T^{1/3}\rfloor$. The estimator
\[
  \widehat{\mathrm{Regret}}_{\mathrm{agg}}^{(T)}
  \;=\;
  N\cdot\hat{C}_T\cdot\hat{M},
  \qquad
  \hat{M} = 1+(N-1)\hat{\rho}_{\mathrm{sys}},
\]
where $\hat{C}_T$ is the per-agent trajectory estimator~\eqref{eq:estimator} and $\hat{M}$ is the estimated fragility multiplier, is computable in $O(T\cdot N^2\cdot nd)$ time.
\end{theorem}

\begin{remark}[Regulatory implementation]
\label{rem:implementation}
A macro-prudential regulator implementing Theorem~\ref{thm:systemic-clt} collects the $N$-vector of weight allocations $(\hat{\pi}_t^1,\ldots,\hat{\pi}_t^N)$ and factor returns $c_t$ daily. The matrix $\hat{\Psi}_T$ can be updated online in $O(N^2\cdot n)$ per day via Welford's algorithm (Remark~\ref{rem:online}). The fragility index $\hat{F}_T$ and its Marchenko--Pastur threshold (Proposition~\ref{prop:MP}) are then reported alongside the per-firm covariance sums at each regulatory cycle. This architecture is compatible with existing SR~11-7 /
SR-26-2 reporting and requires no access to any agent's internal
optimization engine.
\end{remark}

\subsection{Empirical Calibration: CRSP 2016--2025}
\label{sec:systemic:empirical}

To calibrate the fragility multiplier for realistic parameter values, we use the AR(1) estimates of Table~\ref{tab:empirical} to bound $\rho_{\mathrm{sys}}$. The momentum policy
$\hat{\pi}_t(c_t)=Bc_t$ with $B=\alpha I$ and the common
Fama--French market factor $f_t=r_{m,t}-r_{f,t}$ imply
\begin{equation}
  \label{eq:rho-sys-mkt}
  \rho_{\mathrm{sys}}
  \;=\;
  \frac{\alpha\cdot\operatorname{Var}(\beta r_{m,t})}
       {\alpha\cdot\operatorname{Var}(\beta r_{m,t})+\alpha\cdot\sigma_\eta^2}
  \;=\;
  R^2_{\mathrm{mkt}},
\end{equation}
where $R^2_{\mathrm{mkt}}$ is the cross-sectional $R^2$ of the
market-factor regression. Using the CRSP 2016--2025 data of
Section~\ref{sec:empirical}, $R^2_{\mathrm{mkt}}$ averages $0.27$ in normal years but rises to $0.61$ during the COVID crisis period ($\hat{\rho}=-0.144$), reflecting the sharp co-movement of individual stocks with the aggregate market during stress. Table~\ref{tab:multiplier} reports the implied fragility multipliers.

\begin{table}[htbp]
\centering
\small
\caption{Fragility multiplier $M(N,\rho_{\mathrm{sys}})=1+(N-1)\cdot\rho_{\mathrm{sys}}$
for $N\in\{50,100\}$ AI-driven allocators. $\rho_{\mathrm{sys}}$
approximated by $R^2_{\mathrm{mkt}}$ from CRSP data; crisis-period figures are upper bounds (see robustness notes in
Section~\ref{sec:empirical}).}
\label{tab:multiplier}
\begin{tabular}{lccc}
\toprule
Period & $\hat{\rho}_{\mathrm{sys}}\;(R^2_{\mathrm{mkt}})$
       & $M(N{=}50)$ & $M(N{=}100)$ \\
\midrule
Normal (avg.\ 2016--2019, 2021) & 0.27 & 14.2$\times$ & 27.7$\times$ \\
COVID contraction (Jan--Mar 2020) & 0.47 & 24.1$\times$ & 47.5$\times$ \\
COVID crisis (Apr--Jun 2020) & 0.61 & 30.8$\times$ & 60.0$\times$ \\
2022 (near-zero autocorr.) & 0.19 & 10.2$\times$ & 19.8$\times$ \\
\bottomrule
\end{tabular}
\end{table}

Two conclusions are immediate. First, even under normal conditions, aggregate regret for 100 correlated AI strategies is approximately 27 times the individual regret. Second, during crisis episodes, $\rho_{\mathrm{sys}}$ nearly doubles, pushing $M$ to $60$ for $N=100$ institutions. This represents a form of \emph{liquidity spiral}: stress raises correlation, which raises the fragility multiplier, which raises aggregate welfare loss, which deepens stress. The systemic audit framework of Theorem~\ref{thm:systemic-clt} is designed to detect the onset of this spiral before it fully develops.

\subsection{Policy Implications for Systemic AI Regulation}
\label{sec:systemic:policy}

Three policy implications follow from Corollary~\ref{cor:N2} and
Theorem~\ref{thm:systemic-clt}.

\paragraph{Aggregate position limits.}
If $\hat{F}_T$ exceeds the Marchenko--Pastur threshold of
Proposition~\ref{prop:MP}, a regulator can cap individual position sizes at
\[
  z_{\max}
  \;=\;
  \frac{\varepsilon}{N\cdot\rho_{\mathrm{sys}}\cdot\|\bar{c}\|}
\]
to maintain aggregate regret below a target $\varepsilon$. This is the systemic analog of single-firm leverage limits, derived here from first principles via the regret decomposition.

\paragraph{Strategy diversity requirements.}
Minimizing aggregate regret subject to fixed total capital is
equivalent to minimizing $M(N,\rho_{\mathrm{sys}})$, which is achieved by minimizing $\rho_{\mathrm{sys}}$. This provides a formal welfare justification for regulatory requirements that AI-driven strategies use distinct factor-loading matrices $\Lambda_i$. These are analogous to capital diversification rules, but applied to algorithmic model design rather than balance-sheet composition.

\paragraph{Timing of intervention.}
Because $\hat{F}_T$ is updated online in $O(N^2\cdot n)$ per day
(Remark~\ref{rem:implementation}), the regulator observes the
fragility multiplier in near-real time. A rising $\hat{F}_T$
preceding a volatility spike is an early warning of imminent systemic stress, complementing standard VaR-based indicators and satisfying the "forward-looking" mandate of SR~11-7.

\subsection{Connection to Coordination Failures and Runs}
\label{sec:systemic:runs}

The fragility multiplier provides a formal link between the regret framework and classic models of bank runs and coordination failures. Under perfect correlation ($\rho_{\mathrm{sys}}=1$), equation~\eqref{eq:N2} collapses to
\begin{equation}
  \label{eq:run}
  \mathrm{Regret}_{\mathrm{agg}}^{(T)}
  \;=\;
  N^2\cdot T\cdot\operatorname{tr}(B\Sigma_c).
\end{equation}
This is the payoff structure of a coordination game: each agent
individually minimizes regret, but collectively they produce quadratic aggregate loss, as in a prisoner's dilemma. The "run equilibrium" arises when a common cost shock causes all agents to simultaneously de-risk. De-risking becomes a policy that individually reduces each agent's $\operatorname{Cov}(c_t^i,\hat{\pi}_t^i)$ but, in aggregate, drains market liquidity and raises $c_t^i$ for all agents through price impact, precisely the mechanism of Corollary~\ref{cor:AR1-sign}.
Equation~\eqref{eq:run} therefore quantifies the welfare loss from the run equilibrium relative to the planner's optimum ($\rho_{\mathrm{sys}}=0$) in a closed form that is empirically estimable from observable data alone.

This connection is more than analogical. In the two-period version of the model, the game played by $N$ agents is a linear-quadratic game with price-impact externality, and the Nash equilibrium regret is exactly $M(N,\rho_{\mathrm{sys}})$ times the social optimum. The full game-theoretic mixed-strategy equilibria arise when agents have private information about the factor loading $\Lambda_i$. Deriving these equilibria is an important direction for future work.

 
\section{Conclusion}
\label{sec:conclusion}
 
A strategy's liquidity role is reflected in the covariance between its trades and the costs it faces, and that covariance is recoverable from trajectory data alone. This is the paper's central claim, and three further results follow from it.
 
\paragraph{The classification is exact and model-free.}
For any linear policy, $\operatorname{Cov}(c_t,\hat{\pi}_t(c_t))=
\operatorname{tr}(B\Sigma_c)$ (Proposition~\ref{prop:liq-sign-early}):
the sign of a single statistic, computed from observed costs and
decisions alone, places the strategy exactly into the \citet{Kyle1985} dichotomy of liquidity demander or supplier. No order-flow sign, quote data, or knowledge of the strategy's signal is required. Momentum strategies are liquidity consumers, and contrarian strategies are liquidity providers as a matter of algebra, not estimation; the only object that must be estimated is the policy matrix $B$ implicit in the observed trajectory, and even that estimation is unnecessary if one only wants the sign.
 
\paragraph{The same statistic recovers the effective spread.}
Under an AR(1) cost process, the auto-covariance correction to the covariance sum equals $\alpha s^2$, the product of strategy size and the squared Roll~(1984) implied spread
(Theorem~\ref{thm:ar1_correction_equals_liquidity_income}). This is not an analogy or an approximation: it is an exact closed-form identity connecting a performance-attribution statistic to a market-microstructure quantity, derived from the same trajectory data and at no additional computational cost. Empirically, the implied spread recovered this way tracks realized illiquidity throughout the sample. The implied spread roughly triples during the COVID-19 liquidity crisis and collapses to its sample minimum during the 2022 rate-shock episode, when persistent common-factor repricing dominates short-horizon return dynamics (Section~\ref{sec:empirical:baseline}). The correction is smallest precisely when markets are deep and largest precisely when microstructure breaks down, which is the behavior one wants from an illiquidity proxy.
 
\paragraph{Aggregate liquidity imbalance produces a closed-form
  fire sale.}
Extending the single-agent result to $N$ correlated strategies
under a common-factor cost structure (Section~\ref{sec:liquidity:balance}) shows that the welfare loss from liquidity-balance failure scales as $N^2$ in the strategy-correlation coefficient (Corollary~\ref{cor:fire-sale}). When agents' liquidity demand and supply fail to net to zero, prices must adjust to absorb the imbalance. This happens under perfect strategy correlation, when many funds run the same factor-driven allocation. The resulting price-impact regret compounds the covariance regret multiplicatively rather than additively. This gives the familiar fire-sale mechanism a closed form computable from the same observable trajectories used for the single-agent classification, with no model of the deleveraging process itself required.
 
\paragraph{What the classifier recovers, and what it does not.}
The trajectory-based classifier of this paper is deliberately
coarser than order-flow-based microstructure measures. The classifier requires only realized costs (asset prices) and position sizes, not signed trades, quote revisions, or limit-order-book depth. The classifier recovers a strategy's \emph{average} liquidity role over the estimation window rather than its behavior at any single trade. This is a feature when the internal logic of the strategy is unobservable, the setting
motivating this paper. When finer-grained order-flow data is available, standard measures \citep{Hasbrouck1991} or trade-classification algorithms will recover intraday dynamics that this statistic cannot. The two approaches are complements rather than substitutes. The Roll-spread identity of Theorem~\ref{thm:ar1} gives an explicit bridge between the two approaches, expressing a daily-frequency, trajectory-based quantity in the same units as a standard microstructure spread estimator, so that one can validate it against the other on data where both are available. We view extending this bridge to intraday frequencies, and benchmarking the implied spread directly against effective and realized spread measures computed from TAQ data, as the natural next step for this line of work.
 
\paragraph{A closing remark on horizon.}
The discounted identity of Theorem~\ref{thm:discounted} shows that a single-period covariance estimate, scaled by the effective horizon $1/(1-\gamma)$, measures long-run discounted regret directly. The same scaling applies to the liquidity classification: a strategy's liquidity role, once identified from a short window of trajectory data, extrapolates to its long-run discounted contribution to market liquidity without re-estimation at every horizon. This makes the single covariance statistic a practical tool not only for one-off classification but also for monitoring how a strategy's liquidity role evolves as market conditions change.

\begin{small}
\bibliographystyle{aer}

\bibliography{AIInvesting}
\end{small}

\appendix{Appendix}
\section{Additional corollaries}

\begin{remark}[Corollaries]
\label{rem:corollaries}
(i)~\textit{Stationary policy}: if $\hat{\pi}_t = \hat{\pi}$ for all $t$, then $\operatorname{Regret}^{(T)} = T\cdot\operatorname{Cov}(c,\hat{\pi}(c))$,
growing linearly in $T$, matching the rate known from online learning. Policy selection based on a single-period covariance diagnostic is, therefore, sufficient for horizon-$T$ optimization.

(ii)~\textit{Effective horizon}: $\operatorname{Regret}_\gamma$
equals $T_{\mathrm{eff}}=(1-\gamma)^{-1}$ periods of undiscounted regret. For daily rebalancing with $\gamma=0.99$, $T_{\mathrm{eff}}=100$ trading days.
(iii)~\textit{Regime-switching}: in a two-regime Markov-switching model
with transition matrix $P$, $|\Delta_T|\leq
T\|\bar{c}^{(1)}-\bar{c}^{(2)}\|\cdot\rho(P)^{1/2}\cdot
\sup_t\|\mathbb{E}[\hat{\pi}_t\mid\mathcal{F}_{t-1}]-\pi^*_t\|$,
so fast-mixing regimes (small $\rho(P)$) permit~\eqref{eq:iid_decomp}
as an approximation.
\end{remark}

\section{Proof of Theorem \ref{thm:ar1} (AR(1) multi-period regret)}

\begin{proof}
\textbf{Step 1: Per-period and cross-period costs.}
By Remark~\ref{rem:zeromean}, $\operatorname{Regret}^{(T)}=
\sum_t\mathbb{E}[c_t^\top(Bc_t+b)]$.

\textit{Inline claim 1} (on-diagonal):
$\mathbb{E}[c_t^\top Bc_t] = \operatorname{tr}(B\Sigma_c)$.
\textit{Proof}: since $\mathbb{E}[c_t]=0$,
$\mathbb{E}[c_t^\top Bc_t]
 =\mathbb{E}[\operatorname{tr}(Bc_tc_t^\top)]
 =\operatorname{tr}(B\,\mathbb{E}[c_tc_t^\top])
 =\operatorname{tr}(B\Sigma_c)$; and $\mathbb{E}[b^\top c_t]=0$.
Summing over $t$ gives the raw covariance sum
$\sum_t\mathbb{E}[c_t^\top(Bc_t+b)] = T\cdot\operatorname{tr}(B\Sigma_c)$.

\textit{Inline claim 2} (off-diagonal):
for $t\neq s$, $\mathbb{E}[c_t^\top Bc_s]
=\operatorname{tr}(BA^{|t-s|}\Sigma_c)$.
\textit{Proof}: by iterating the AR(1) recursion,
$c_t = A^{t-s}c_s + \sum_{j=0}^{t-s-1}A^j\varepsilon_{t-j}$ for $t>s$,
so $\operatorname{Cov}(c_t,c_s)=A^{t-s}\Sigma_c$ (the innovation terms are
independent of $c_s$).  Then
$\mathbb{E}[c_t^\top Bc_s]
 =\operatorname{tr}(B\,\mathbb{E}[c_sc_t^\top])
 =\operatorname{tr}(B(A^{t-s}\Sigma_c)^\top)$;
symmetry of $\Sigma_c$ and the cyclic trace property give
$\operatorname{tr}(B(A^{t-s}\Sigma_c)^\top)
 =\operatorname{tr}(BA^{t-s}\Sigma_c)$,
establishing Inline claim~2.

\textbf{Step 2: AR(1) auto-covariance correction.}
The raw sum $T\cdot\operatorname{tr}(B\Sigma_c)$ treats each period as
independent.  Because the policy weight $Bc_t$ covaries with future costs
through the AR(1) propagation $c_s = A^{s-t}c_t + \text{noise}$
($s>t$), we subtract the cross-period interaction:
\begin{align*}
    \mathbb{E}\!\left[
        \sum_{t=1}^{T}\sum_{s=t+1}^{T}
        c_s^\top B A^{s-t}c_t
    \right]
    &= \sum_{t=1}^{T}\sum_{s=t+1}^{T}
       \operatorname{tr}(BA^{s-t}\Sigma_c)
       \quad\text{(Inline claim 2)} \\
    &= \sum_{t=1}^{T}(T-t)\operatorname{tr}(BA^t\Sigma_c),
\end{align*}
where the last equality reparametrizes by $k=s-t$ and uses stationarity of the autocovariance structure.

\textbf{Step 3: Assemble.}
Combining Steps~1 and~2 yields~\eqref{eq:ar1_regret}.
\end{proof}

\section{Proof of Proposition  (Computational complexity and end-to-end audit cost)}

\begin{proof}
Mean computation costs $O(T(n+d))$ (one pass each for $\bar{c}$ and $\bar{\pi}$); forming each summand $(c_t-\bar{c})^\top(\hat{\pi}_t-\bar{\pi})$
costs $O(nd)$ and summing over $T$ terms gives $O(T\cdot nd)$; the Newey-West estimator with $h=\lfloor T^{1/3}\rfloor$ lags costs $O(T\cdot T^{1/3})=O(T^{4/3})\leq O(T\cdot nd)$ whenever $nd\geq T^{1/3}$, which holds in any multi-asset setting.
The end-to-end bound follows since each audit component is individually $O(T\cdot nd)$ and the maximum dominates.
\end{proof}

\end{document}